\def\year{2020}\relax
\documentclass[letterpaper]{article} % DO NOT CHANGE THIS
\usepackage{aaai20}  % DO NOT CHANGE THIS
\usepackage{lmodern}
\usepackage[T1]{fontenc}
\usepackage{cancel}
\usepackage{mathrsfs}  
\usepackage[hyphens]{url}  % DO NOT CHANGE THIS
\usepackage{graphicx} % DO NOT CHANGE THIS
\usepackage{graphicx}  % DO NOT CHANGE THIS
\newcommand{\aT}{\mathsf{T}}
\newcommand{\Perp}{\perp \! \! \! \perp}
\usepackage{amsthm}
\usepackage{amssymb}
\usepackage{amsmath}
\usepackage{tikz}
\usetikzlibrary{arrows}
\usepackage{array,enumitem}
\theoremstyle{definition}
\newtheorem{theorem}{Thm.}
\newtheorem{lemma}[theorem]{Lem.}
\newtheorem{corollary}[theorem]{Cor.}
\newtheorem{proposition}[theorem]{Prop.}
\newtheorem{definition}{Def.}

\newtheorem{fact}{Fact}

\usepackage{stmaryrd}
\newcommand{\semantics}[1]{\big\llbracket{#1}\big\rrbracket}
\usepackage{latexsym}

\newcommand{\Val}{\textnormal{Val}}
\newcommand{\Ll}{\mathcal{L}}
\newcommand{\base}{\textnormal{base}}
\newcommand{\Lbasei}{\Ll^{\base}_i}
\newcommand{\Lbase}{\Ll^{\base}}
\newcommand{\Lbasethree}{\Ll^{\base}_3}

\newcommand{\Lcond}{\mathcal{L}_{\textnormal{cond}}}
\newcommand{\Lprop}{\mathcal{L}_{\textnormal{prop}}}
\newcommand{\Lfull}{\mathcal{L}_{\textnormal{full}}}

\newcommand{\Lint}{\mathcal{L}_{\textnormal{int}}}
\newcommand{\AX}{\textsf{AX}}

\newcommand{\tT}{\mathbf{t}}

\newcommand{\Poly}{\textsf{Poly}}

\newcommand{\AddComm}{\textsf{AddComm}}
\newcommand{\AddAssoc}{\textsf{AddAssoc}}

\newcommand{\One}{\textsf{One}}

\newcommand{\0}{\underline{0}}
\newcommand{\1}{\underline{1}}

\newcommand{\EndoInst}{\mathbf{v}}
\newcommand{\Unit}{\mathbf{u}}
\newcommand{\Units}{\mathfrak{u}}

\newcommand{\Vphi}{\mathbf{V}_{\varphi}}
\newcommand{\Lfullphi}{\mathcal{L}_{{\mathrm{full}(\varphi)}}}
\newcommand{\Lintphi}{\mathcal{L}_{\mathrm{int}(\varphi)}}
\newcommand{\Lcondphi}{\mathcal{L}_{\mathrm{cond}(\varphi)}}

\newcommand{\influences}{{\leadsto}}
\usepackage{centernot}
\newcommand{\ninfluences}{{\centernot{\leadsto}}}

\newcommand{\influencesindirect}{{\rightsquigarrow}}

\makeatletter
\providecommand{\leftsquigarrow}{%
  \mathrel{\mathpalette\reflect@squig\relax}%
}
\newcommand{\reflect@squig}[2]{%
  \reflectbox{$\m@th#1\rightsquigarrow$}%
}
\makeatother

\usepackage{color}

\newcommand{\ssep}{{\,:\,}}
\newcommand{\notmodels}{\centernot{\models}}

\DeclareMathOperator{\supp}{supp}
\DeclareMathOperator{\cons}{sat}

\renewcommand{\epsilon}{\varepsilon}

\usepackage{comment}

\title{Probabilistic Reasoning across the Causal Hierarchy\thanks{In \emph{Proceedings of the Thirty-Fourth AAAI Conference on Artificial Intelligence} (AAAI-20).}}
\author{Duligur Ibeling,\textsuperscript{\rm 1} Thomas Icard\textsuperscript{\rm 2}\\% All authors must be in the same font size and format. Use \Large and \textbf to achieve this result when breaking a line
\textsuperscript{\rm 1}Department of Computer Science, Stanford University\\
\textsuperscript{\rm 2}Department of Philosophy, Stanford University\\
duligur@stanford.edu, icard@stanford.edu% email address must be in roman text type, not monospace or sans serif
}

\begin{document}

\maketitle

\begin{abstract}
We propose a formalization of the three-tier causal hierarchy of association, intervention, and counterfactuals as a series of probabilistic logical languages. Our languages are of strictly increasing expressivity, the first capable of expressing quantitative probabilistic reasoning---including conditional independence and Bayesian inference---the second encoding $do$-calculus reasoning for causal effects, and the third capturing a fully expressive $do$-calculus for arbitrary counterfactual queries. We give a corresponding series of finitary axiomatizations complete over both structural causal models and probabilistic programs, and show that satisfiability and validity for each language are decidable in polynomial space.
\end{abstract}

\section{Introduction and Summary}
Intelligence commonly involves prediction, anticipating future events on the basis of past observations (e.g., ``Will the water pipes freeze again this winter?''). Intelligent planning and decision-making additionally require predicting what \emph{would} happen under a hypothetical action (``Will the pipes freeze if we keep the heat on?''). An even more sophisticated ability---critical for tasks like explanation---is to reason counterfactually about what \emph{would have} happened given knowledge about what in fact happened (``Would the pipes have frozen if we had left the heat on, given that the heat was off and the pipes in fact froze?''). These three modes of reasoning constitute a \emph{causal hierarchy} \cite{Shpitser,Pearl2009}, highlighting the significance of structural causal knowledge for flexible thought and action.

The aim of the present article is to gain conceptual as well as technical insight into this hierarchy by employing tools from logic. Loosely following earlier work \cite{Shpitser}, we propose a characterization of its levels in terms of logical syntax. Semantically, all three languages are interpreted over the same class of models, namely structural causal models (and later probabilistic programs). The languages differ in how much they can express about these models. $\mathcal{L}_1$, the language of association, expresses only ``pure'' probabilistic facts and relationships; $\mathcal{L}_2$, the language of probabilistic intervention, allows expressing probabilities of basic conditional ``if\dots then\dots'' statements; $\mathcal{L}_3$, the language of probabilistic counterfactuals, encodes probabilities for arbitrary boolean combinations of such conditional statements. Using standard ideas from logic and existing results, we can address questions about definability and expressiveness. For instance, it is easy to prove in our framework that each language is strictly more expressive than those below it in the hierarchy (Prop. \ref{l1l2}, \ref{l2l3} below). We can also interpret well-known insights from the graphical models and causal learning literatures as \emph{graph definability} results for appropriate probabilistic logical languages, analogously to correspondence theory in modal logic \cite{vanBenthem2001}.

In possession of a precise syntax and semantics for probabilistic causal reasoning, questions of axiomatization naturally arise. That is, we would like to identify a perspicuous set of basic principles that underly all such reasoning. One of our main technical contributions is a series of finitary (sound and complete) axiomatizations for each level of the causal hierarchy (Thm. \ref{thm:soundcomplete}), relying on methods from semialgebraic geometry. As a corollary to these completeness results, we also reveal a ``small-model property'' with the consequence that satisfiability and validity for $\mathcal{L}_1$, $\mathcal{L}_2$, and $\mathcal{L}_3$ can be decided in polynomial space (Thm. \ref{thm:complexity}).

Finally, in the last part of the paper we consider an alternative interpretation for our three logical languages. Probabilistic programs, with an appropriate notion of \emph{causal intervention}, provide a procedural semantics for probabilistic counterfactual claims and queries. We establish an equivalence between these models and a natural subclass of \emph{computable} structural causal models (Thm. \ref{thm:equivalence}). The equivalence in turn implies soundness and completeness of our axiomatizations with respect to this interpretation as well.

\subsection{Relation to Previous Work}
While deterministic causal counterfactuals and probabilistic logics have both received extensive treatment independently, the present contribution appears to be the first systematic study of probabilistic counterfactuals. Our work thus synthesizes and improves upon a long line of previous work. Axioms for causal conditionals interpreted over structural causal models are well understood \cite{GallesPearl,Halpern2000,Pearl2009,Zhang,II2019} and play a distinct role in causal reasoning tasks such as identification \cite{Shpitser,Bareinboim}. Indeed, some prominent approaches to causal learning and identification  even employ algorithmic logic-based techniques \cite{Hyttinen1,Hyttinen,Trian}.

Meanwhile, much is known about formalized probability calculi. \cite{Fagin} considered a probability logic built over a language of polynomials sufficiently expressive to encompass essentially all ordinary probabilistic reasoning about propositional facts, including Bayesian inference, conditional independence, and so on. However they left open the problem of explicit axiomatization.
A (strongly) complete axiomatization was later provided by \cite{Perovic} using an infinitary proof rule. Whereas our main interest is causal reasoning beyond the first level of the hierarchy, Thm. \ref{thm:soundcomplete} incidentally establishes the first (weakly) complete \emph{finitary} axiomatization for ``pure'' probability logic over a language of polynomials. Moving to the second and third levels of the hierarchy, Thm. \ref{thm:soundcomplete} also presents the first combined axiomatization for probabilistic reasoning about causal counterfactuals.
{At the second level, we draw upon an existing characterization by \cite{TianKP06}.}

\cite{Fagin} established a complexity upper-bound ($\mathsf{PSPACE}$) for their satisfiability problem. On all three levels of the hierarchy, we obtain the same upper bound for our decision problem (Thm. \ref{thm:complexity}).
Both arguments rely crucially on the procedure given by \cite{Canny} to decide the existential theory of a real closed field.
It has been previously suggested to apply cylindrical algebraic decomposition (which decides the full first-order theory of real closed fields) to causal questions \cite{Geiger:1999:QES:2073796.2073822}.

Encoding causal knowledge in an implicit way via a generative probabilistic program has been explored recently by a number of research groups \cite{Lake2017,Pyro,Tavares}. Deterministic conditionals over ``simulation programs'' have been axiomatized \cite{II2018}, showing that in general such an interpretation validates strictly fewer principles than structural causal models. This weaker axiomatic system was also embedded in a probability logic with linear inequalities \cite{Ibeling18}. It is possible, however, to restrict the class of probabilistic programs so as to ensure equivalence with (an appropriate class of) structural causal models \cite{II2019}. We draw on all of this work in what follows.

\section{Structural Models}
We are interested in \emph{structural causal models} (see, e.g., \cite{Pearl2009}) defined over a signature consisting of a fixed set $\mathbf{V}$ of \emph{endogenous} variables.
Each model is also defined over a set $\mathbf{U}$ of \emph{exogenous} variables.
Every $R \in \mathbf{V} \cup \mathbf{U}$ takes on a value from an admissible set $\text{Val}(R)$. While $\mathbf{V}$ and $\mathbf{U}$ may be infinite, we assume $\text{Val}(V)$ is finite for all $V \in \mathbf{V}$. %, but we do not assume $\mathbf{V}$ or $\mathbf{U}$ is finite. % TODO: not for exogenous ones.
Given a set of variables $\mathbf{S} \subseteq \mathbf{V} \cup \mathbf{U}$, we call an assignment $\mathbf{s} \in \Val(\mathbf{S})$ of each variable $S \in \mathbf{S}$ to a value $\mathbf{s}(S) \in \text{Val}(S)$ an \emph{instantiation} of $\mathbf{S}$.
\begin{definition}[Structural Causal Model]\label{def:scm}
We define a SCM to be a pair $\mathfrak{M} = (\mathcal{F},P)$, where $P$ is a probability measure on a $\sigma$-algebra $\Sigma$ of instantiations $\mathbf{u}$ of $\mathbf{U}$, and $\mathcal{F} = \{f_V\}_{V\in\mathbf{V}}$ is a set of functions $f_V: \Val(\mathbf{V} \cup \mathbf{U}) \to \Val(V)$, one for each endogenous variable.
\end{definition}
Thus every $f_V$ is a function from endogenous instantiations $\mathbf{v}$ and exogenous instantiations $\mathbf{u}$ to a value $f_V(\mathbf{v}, \mathbf{u}) \in \text{Val}(V)$.
% We call an SCM \emph{recursive} if the following condition holds.
Fixing $\mathbf{u}$, the simultaneous equations $\mathcal{F}(\mathbf{u})$ over $\mathbf{V}$ may have any number of solutions.
In order to guarantee unique solutions, we henceforth consider only \emph{recursive}\footnote{Recursiveness is sufficient though not necessary for unique existence: see \cite[Ex.~2.1]{Halpern2000}.} SCMs in the following sense:\begin{definition}%[Recursive SCM]
\label{def:recursivescm}
$\mathfrak{M}$ is recursive if there is a well-order $\prec$ on $\mathbf{V}$ such that $f_V(\mathbf{v}_1, \mathbf{u}) = f_V(\mathbf{v}_2, \mathbf{u})$ for any $V$, $\mathbf{u}$, and $\mathbf{v}_1, \mathbf{v}_2$ such that $\mathbf{v}_1\big(V'\big) = \mathbf{v}_2\big(V'\big)$ for all $V' \prec V$.
% That is, (the function for) any variable depends only on its predecessors relative to the order.
\end{definition}
Def.~\ref{def:recursivescm} generalizes prior notions of recursiveness, e.g., \emph{open-universe} SCMs \cite[Def.~1]{II2019} where $\prec$ has order type $\omega$, or the strictly finite setting of \cite{Halpern2000}.
Transfinite induction shows $\mathcal{F}(\mathbf{u})$ has a unique solution for any $\mathbf{u}$, and
$P(\mathbf{U})$ thus defines the obvious joint probability distribution $P_{\mathfrak{M}}(\mathbf{V})$.
% \tb{When $\Val(\mathbf{U})$ is discrete, exogenous instantiations $\mathbf{u}$ are often called \emph{units}.}
The same holds under any causal intervention, defined standardly \cite{Spirtes,Pearl2009}:
\begin{definition}[Intervention]
  \label{defn:scmintervention}
An intervention is a partial function $i : V \mapsto \text{Val}(V)$. It specifies variables $\mathrm{dom}(i) \subseteq \mathbf{V}$ to be held fixed, and the values to which they are fixed. Intervention $i$ induces a mapping of SCMs, also denoted $i$, so that $i(\mathfrak{M})$ is identical to $\mathfrak{M}$, but with $f_V$ replaced by the constant function $f_V(\cdot) = i(V)$ for each $V \in \mathrm{dom}(i)$. We say $i$ is finite whenever $\mathrm{dom}(i)$ is finite.
\end{definition}

In order to establish a correspondence (Thm. \ref{thm:equivalence}) with probabilistic programs, we introduce a further restriction on SCMs, following previous work on computable causal models \cite{IM,II2019}. Def. \ref{compmodel} below assumes a very simple ``coin-flip'' probability space; we leave it as an exercise to show that there is no loss of generality compared to using any computable probability space as in, e.g., \cite{Ackerman2019}, including any standard continuous probability distributions.
\begin{definition} \label{compmodel}
  $\mathfrak{M} = (\mathcal{F}, P)$ is \emph{computable} if (1) its exogenous variables consist of infinitely many binary $U_1,U_2,\dots$, i.i.d. and uniform under $P$, and (2)
  the collection $\mathcal{F} = \{f_V\}_{V}$ is uniformly computable \cite{Weihrauch}.
\end{definition}
Call a model $\mathfrak{M} = (\mathcal{F}, P)$ \emph{measurable} if under every finite intervention $i$, the joint distribution $P_{i(\mathfrak{M})}(\mathbf{V})$ is well-defined. The next Fact is straightforward.
\begin{fact}
  \label{prop:measurable}
Every computable SCM is measurable.
\end{fact}
\begin{proof}
Let $\mathbf{v} \in \Val(\mathbf{V})$ and let $\Units = \bigcap_{V \in \mathbf{V}} \Units_V$ where $\Units_V = \big\{ \Unit : f_V(\EndoInst, \Unit)  = \EndoInst(V)\big\}$. Letting $\Sigma$ be our $\sigma$-algebra, $P_{i(\mathfrak{M})}(\EndoInst)$ is well-defined if $\Units \in \Sigma$. It suffices to show that $\Units_V \in \Sigma$ for all $V$ since $\Sigma$ closes under countable intersection. There is a machine that halts outputting the value $f_V(\EndoInst, \Unit)$ for any $\Unit \in \Units_V$. By then it has seen only finitely many exogenous bits, whose values we write in a finite vector $\vec{u}(\Unit)$.
Thus the cylinder set $\Units'\big(\vec{u}(\Unit)\big)$ of $\Unit'$ that agree with $\vec{u}(\Unit)$ wherever the latter is defined is contained in $\Units_V$. So letting $\Units'_V = \bigcup_{\Unit \in \Units_V} \Units'\big(\vec{u}(\Unit)\big)$, we have $\Units'_V \subseteq \Units_V$. Every cylinder is in $\Sigma$ and there are only countably many cylinders, so $\Units'_V \in \Sigma$; obviously $\Units_V \subseteq \Units'_V$ so $\Units_V \in \Sigma$ as desired.
\end{proof}
Let $\mathcal{M}$ be the class of measurable, recursive SCMs.\footnote{Any two models in $\mathcal{M}$ have the same endogenous signature $\mathbf{V}$. But they do not necessarily share the same set of exogenous variables $\mathbf{U}$; SCMs are thus \emph{non-parametric}.}
Further, let $\mathcal{M}^*$ be the computable subclass of $\mathcal{M}$.

\subsection{Causal Influence}

We now define a relation $\influences$ of \emph{direct causal influence} (cf. \cite{10.1093/bjps/axi147}).
% \begin{definition}[Direct influence]\label{def:directcausalinfluence}
% Let $\mathfrak{M}$ be an SCM (Def.~\ref{def:scm}).
Say $X \influences_{\mathfrak{M}} Y$ if there are endogenous instantiations $\mathbf{v}_1, \mathbf{v}_2 \in \Val(\mathbf{V})$ differing only at $X$ and some non-negligible set $\Units \subseteq \Val(\mathbf{U})$ such that
$f_Y(\mathbf{v}_1, \mathbf{u}) = f_Y(\mathbf{v}_1, \mathbf{u}') \neq f_Y(\mathbf{v}_2, \mathbf{u}) = f_Y(\mathbf{v}_2, \mathbf{u}')$ for any $\mathbf{u}, \mathbf{u}' \in \Units$.
% \end{definition}
This in turn induces a directed graph
$\mathcal{G}_\mathfrak{M} = (\mathbf{V},\influences_{\mathfrak{M}})$.
The following observation shows that recursiveness generalizes the ubiquitous \emph{semi-Markov}
(i.e., finite acyclicity) assumption.
Fix a finite $\mathbf{V}_{\centernot{\infty}}$ and let $\mathcal{M}_{\centernot{\infty}}$ be the class of SCMs over $\mathbf{V}_{\centernot{\infty}}$ in which $P(\mathbf{u}) > 0$ for every $\mathbf{u}$.
Then:
\begin{proposition}
\label{prop:semimarkovisrecursive}
If $\mathfrak{M} \in \mathcal{M}$ is recursive, then $\mathcal{G}_\mathfrak{M}$ is well-founded.
The converse holds if $\mathfrak{M} \in \mathcal{M}_{\centernot{\infty}}$.
\end{proposition}
\begin{proof}
If $\mathfrak{M}$ is recursive under $\prec$, then $X \influences_{\mathfrak{M}} Y$ entails $X \prec Y$.
Converse: we claim that $\mathfrak{M}$ is recursive under any topological order $\prec$ of the dag $\mathcal{G}_\mathfrak{M}$.
Suppose not so that there are $V$, $\mathbf{u}$ and $\mathbf{v}_1$, $\mathbf{v}_2$ differing only on $\mathbf{V}' = \{ V' : V' \succeq V\}$ such that $f_V(\mathbf{v}_1, \mathbf{u}) \neq f_V(\mathbf{v}_2, \mathbf{u})$.
Since $\mathbf{V}'$ is finite,
there is a sequence $\mathbf{v}^0, \dots, \mathbf{v}^n$
such that $\mathbf{v}^0 = \mathbf{v}_1$, $\mathbf{v}^n = \mathbf{v}_2$, and each $\mathbf{v}^i$ differs from $\mathbf{v}^{i-1}$ at a single variable (cf. \cite[Prop.~1]{II2019}).
Let $i$ be the first such that $f_V(\mathbf{v}^i, \mathbf{u}) \neq f_V(\mathbf{v}^{i-1}, \mathbf{u})$ and $V^i$ be that variable at which $\mathbf{v}^i, \mathbf{v}^{i-1}$ differ.
Then $V^i \influences_{\mathfrak{M}} V$ while $V \preceq V^i$, a contradiction.
\end{proof}
%Throughout we restrict attention to SCMs $\mathfrak{M}$ such that %$\mathcal{G}_\mathfrak{M}$
%$\influences_{\mathfrak{M}}$
%is well-founded,

We say $\mathfrak{M}$ is \emph{Markov} if each variable is $P_{\mathfrak{M}}$-independent of its non-descendants (in $\mathcal{G}_{\mathfrak{M}}$) conditional on its parents.
The Markov condition is guaranteed provided the exogenous variables are jointly independent and every endogenous variable depends only on one exogenous variable \cite[Thm. 1.4.1]{Pearl2009}. We define  \emph{d-separation} on a dag % directed graph
$\mathcal{G}$ standardly, and write $\big((\mathbf{X}\Perp\mathbf{Y})|\mathbf{Z}\big)_\mathcal{G}$ to say that the variables $\mathbf{X}$ are d-separated from $\mathbf{Y}$ given $\mathbf{Z}$. In Markov structures d-separation guarantees conditional independence.

\section{Probabilistic Conditionals}
\subsection{Syntax}
We define a succession of language fragments as follows, where $V \in \mathbf{V}$ and $v \in \Val(V)$:
\begin{eqnarray*}
 \mathcal{L}_{\textnormal{int}} & ::= & \top \quad | \quad V=v \quad | \quad \mathcal{L}_{\textnormal{int}} \wedge  \mathcal{L}_{\textnormal{int}} \\
 \mathcal{L}_{\textnormal{prop}} & ::= & V=v \quad | \quad \neg \mathcal{L}_{\textnormal{prop}} \quad | \quad  \mathcal{L}_{\textnormal{prop}} \wedge \mathcal{L}_{\textnormal{prop}} \\
  \mathcal{L}_{\textnormal{cond}} & ::= & [\mathcal{L}_{\textnormal{int}}]\mathcal{L}_{\textnormal{prop}} \\
\mathcal{L}_{\textnormal{full}} & ::= & \mathcal{L}_{\textnormal{cond}} \quad | \quad \neg \mathcal{L}_{\textnormal{full}} \quad | \quad  \mathcal{L}_{\textnormal{full}} \wedge \mathcal{L}_{\textnormal{full}}
\end{eqnarray*}
Based on these fragments we define a sequence of three increasingly expressive probabilistic languages $\{\mathcal{L}_i\}_{i = 1, 2, 3}$. Each language $\mathcal{L}_i$ speaks about probabilities over a \emph{base language} $\mathcal{L}_i^{\textnormal{base}}$. The base languages are
\[
  \mathcal{L}_1^{\textnormal{base}} = \Ll_{\textnormal{prop}}, \quad \mathcal{L}_2^{\textnormal{base}} = \Ll_{\textnormal{cond}}, \quad \mathcal{L}_3^{\textnormal{base}} = \Ll_{\textnormal{full}}.
\]
Our languages describe facts about the probabilities that base language formulas hold. As our formulas are finitary, such facts correspond to polynomials in these probabilities. Let us make this precise.
Fixing a terminal set $T$, define the \emph{polynomial terms in the variables $T$} to be those $\mathbf{t}$ generated by this grammar (where $T$ generates any element of $T$):
\begin{eqnarray*}
\mathbf{t} & ::= & T \quad | \quad \mathbf{t} + \mathbf{t} \quad | \quad \mathbf{t} \cdot \mathbf{t} \quad | \quad - \mathbf{t}.
\end{eqnarray*}
Then the terms of $\Ll_i$ are
polynomials over probabilities of base formulas, i.e., polynomial terms in the variables $\big\{ \mathbb{P}(\epsilon) : \epsilon \in \Lbasei \big\}$. The language $\mathcal{L}_{i}$, $i = 1,2,3$, is then a propositional language of term inequalities:
\begin{eqnarray*}
 \mathcal{L}_{i} & ::= & \mathbf{t} \geqslant \mathbf{t} \quad | \quad \neg  \mathcal{L}_{i} \quad | \quad  \mathcal{L}_{i} \wedge  \mathcal{L}_{i}
 \end{eqnarray*}
 where $\mathbf{t} $ is a term of $\Ll_i$.
We employ the following abbreviations. For $\mathcal{L}_{\textnormal{prop}}$ and $\mathcal{L}_{\textnormal{full}}$ we take $\bot$ to stand for any propositional contradiction, and $\top$ for any propositional tautology. For terms: $\underline{0}$ for $\mathbb{P}(\bot)$, $\underline{1}$ for $\mathbb{P}(\top)$. For $\Ll_i$ formulas, we write $\mathbf{t}_1 \equiv \mathbf{t}_2$ for $(\mathbf{t}_1 \geqslant \mathbf{t}_2) \wedge (\mathbf{t}_2 \geqslant \mathbf{t}_1)$, and $\mathbf{t}_1 > \mathbf{t}_2$ for $(\mathbf{t}_1 \geqslant \mathbf{t}_2) \wedge \neg (\mathbf{t}_2 \geqslant \mathbf{t}_1)$. Note that we may use any rational number as a term via representing its numerator as a sum of $\underline{1}$s and clearing its denominator through an inequality $\mathbf{t} \geqslant \mathbf{t}$, and we write $\underline{q}$ for a rational $q$ thus considered as a term.

Strictly speaking, $\mathbb{P}(\beta)$ for $\beta \in \Ll_{\textnormal{prop}}$ is not a well-formed term in $\Ll_2$ or $\Ll_3$. We nonetheless use this notation with the understanding that $\mathbb{P}(\beta)$ is a shorthand for $\mathbb{P}([\top]\beta)$. $\Ll_2$ and $\Ll_3$ thus extend $\Ll_1$ in this sense.

$\Ll_1$, $\Ll_2$, and $\Ll_3$ correspond to the three levels of the causal hierarchy as proposed by \cite{Shpitser,Pearl2009}; see also \cite{BCII}. $\Ll_1$ is simply the language of probability, capturing statements like $\mathbb{P}(Y=y | X=x) \geqslant \underline{1/2}$, which is shorthand for $\big(\mathbb{P}(\top)+\mathbb{P}(\top)\big)\cdot \mathbb{P}(X=x \,\wedge\, Y=y) \geqslant  \mathbb{P}(X=x)$. $\Ll_2$ encompasses assertions about so-called \emph{causal effects}, e.g., statements like $\mathbb{P}([X=x]Y=y) \geqslant \underline{q}$.

\subsection{Semantics}
A model is simply a measurable SCM $\mathfrak{M} = (\mathcal{F},P)$. Since $\mathfrak{M}$ is well-founded, each $\Unit$ determines the values of all endogenous variables. Thus, for $\beta \in \mathcal{L}_{\textnormal{prop}}$ we will write $\mathcal{F},\Unit \models \beta$, defined in the obvious way.
For $\alpha \in \mathcal{L}_{\textnormal{int}}$ we define the intervention operation $i_\alpha$ so that $i_\alpha(\mathcal{F})$ is the result of applying the interventions specified by $\alpha$ to $\mathcal{F}$ (Def. \ref{defn:scmintervention}). Then $\mathcal{F},\Unit \models [\alpha]\beta$ just in case $i_{\alpha}(\mathcal{F}),\Unit \models \beta$. We have thus defined $\mathcal{F}, \Unit \models \epsilon$ for all $\epsilon \in \Lfull$.

% Finally,
% for any $\varphi \in \mathcal{L}$ and
Next, for any $\mathfrak{M} = (\mathcal{F}, P)$ define the set $S_{\mathfrak{M}}(\epsilon) = \{\mathbf{u}: \mathcal{F}, \mathbf{u} \models \epsilon\}$. Measurability of $\mathfrak{M}$ guarantees that $S_{\mathfrak{M}}(\epsilon)$ is always measurable. Toward specifying the semantics of $\mathcal{L}_i$ we define $\semantics{\mathbf{t}}_{\mathfrak{M}}$ recursively, with the crucial clause given by $\semantics{\mathbb{P}(\epsilon)}_{\mathfrak{M}} = P\big(S_{\mathfrak{M}}(\epsilon)\big)$. Satisfaction of $\varphi \in \mathcal{L}_i$  is as expected: $\mathfrak{M} \models \mathbf{t}_1 \geqslant \mathbf{t}_2$ iff $\semantics{\mathbf{t_1}}_{\mathfrak{M}} \geq \semantics{\mathbf{t_2}}_{\mathfrak{M}}$, $\mathfrak{M} \models \neg \varphi$ iff $\mathfrak{M} \notmodels \varphi$, and $\mathfrak{M} \models \varphi \wedge \psi$ iff $\mathfrak{M} \models \varphi$ and $\mathfrak{M} \models \psi$.
As in previous work, it is easy to see that none of the languages  $\mathcal{L}_1,\mathcal{L}_2,\mathcal{L}_3$ is compact \cite{Perovic,II2018,II2019}.
Consequently Thm. \ref{thm:soundcomplete} claims weak completeness only.

\subsection{Comparing Expressivity}
With a precise semantic interpretation of our three languages
in hand, we can now show rigorously that they form a strict hierarchy, in the sense that models may be distinguishable only by moving up to higher levels of the hierarchy.
\begin{proposition}\label{l1l2} $\mathcal{L}_2$ is strictly more expressive than $\mathcal{L}_1$. \end{proposition}
\begin{proof}
Consider $\mathfrak{M}_1$ with $U \sim \textsf{Bernoulli}(0.5)$ and $X := U$ while $Y := X$; in  $\mathfrak{M}_2$ we have $Y := U$ and $X := Y$. It is easy to see by an induction on terms in $\mathcal{L}_1$ that $\semantics{\mathbf{t}}_{\mathfrak{M}_1} = \semantics{\mathbf{t}}_{\mathfrak{M}_2}$,
and thus $\mathfrak{M}_1$ and $\mathfrak{M}_2$ validate the same $\mathcal{L}_1$ formulas. Yet, $\mathfrak{M}_1 \models \mathbb{P}([X=1]Y=1) \equiv \underline{1}$, while $\mathfrak{M}_2 \notmodels \mathbb{P}([X=1]Y=1) \equiv \underline{1}$.  In particular the schema $\mathbb{P}(\beta | \alpha) \equiv \mathbb{P}([\alpha]\beta)$ is also falsified by $\mathfrak{M}_2$, a reflection of the distinction between observation and intervention.
\end{proof}
\begin{proposition}\label{l2l3} $\mathcal{L}_3$ is strictly more expressive than $\mathcal{L}_2$. \end{proposition}
\begin{proof} Consider an example adapted from \cite{Avin} with two endogenous and two exogenous variables $X,U_X,Y,U_Y$, where $U_X \sim \textsf{Bernoulli}(0.5)$ and $X = U_X$, while $U_Y \sim \textsf{Unif}(0,1,2)$. The difference between models $\mathfrak{M}_1$ and $\mathfrak{M}_2$ is the function for binary variable $Y$. In $\mathfrak{M}_1$ we have $Y$ equal to $(X \leftrightarrow U_Y= 0)$, and in $\mathfrak{M}_2$ we have $Y$ given by $(X \rightarrow U_Y=0) \wedge (U_Y=2 \rightarrow X)$. It is then easy to check (by induction) that $\mathfrak{M}_1$ and $\mathfrak{M}_2$ validate all the same $\mathcal{L}_2$ formulas, whereas, e.g., $\semantics{\mathbf{p}}_{\mathfrak{M}_1} \neq \semantics{\mathbf{p}}_{\mathfrak{M}_2}$, with $\mathbf{p}$ the term denoting the probability of necessity and sufficiency $\mathbb{P}([X=0]Y=0 \wedge [X=1]Y=1)$.
\end{proof}
Note also that $\semantics{\mathbb{P}([\top]Y=1 \wedge [X=1]Y=1)}_{\mathfrak{M}_1} \neq \semantics{\mathbb{P}([\top]Y=1 \wedge [X=1]Y=1)}_{\mathfrak{M}_2}$ in this second example, showing that even allowing simple conjunctions of the form $\gamma \wedge [\alpha]\beta$ would increase the expressive power of $\mathcal{L}_2$. It is thus not possible in general to reason in $\mathcal{L}_2$ about  conditional expressions such as $\mathbb{P}([\alpha]\beta|\gamma)$. On the other hand $\mathcal{L}_2$ does handle \emph{conditional effects}, since, e.g., $\mathbb{P}([\alpha]\beta|[\alpha]\gamma) \geqslant \mathbf{t}$ can be rewritten as $\mathbb{P}\big([\alpha](\beta\wedge\gamma)\big) \geqslant \mathbf{t}\cdot \mathbb{P}([\alpha]\gamma)$.\footnote{In the notation of \emph{do}-calculus \cite{Pearl1995,Pearl2009}, the expression $\mathbb{P}([\alpha]\beta|[\alpha]\gamma)$ would be written as $\mathbb{P}(\beta \,|\, do(\alpha),\gamma)$.}

In a companion article we improve upon Props. \ref{l1l2} and \ref{l2l3} by showing that for $i<j$ the $\mathcal{L}_i$-theory of a model \emph{almost-never} (i.e., with measure zero) determines its $\mathcal{L}_j$-theory \cite[Thm. 1]{BCII}.

\subsection{Graph Definability and $Do$-Calculus}

Given the languages and interpretation considered so far, we mention as an aside that it may be enlightening to consider a notion of \emph{graph validity}, analogous to ``frame validity'' in modal logic \cite{vanBenthem2001}. Let us say $\mathcal{G} \models \varphi$ just in case $\mathfrak{M} \models \varphi$ for all Markov structures $\mathfrak{M}$ such that $\mathcal{G}=\mathcal{G}_{\mathfrak{M}}$.

For any dag $\mathcal{G}$ there is a probability distribution $P$ whose conditional independencies are exactly those implied by d-separation in $\mathcal{G}$ \cite{Geiger}. It is then easy to construct an SCM $\mathfrak{M}$ with $\mathcal{G}=\mathcal{G}_{\mathfrak{M}}$ and $P=P_{\mathfrak{M}}$, which immediately gives:\footnote{$\mathbb{P}(\mathbf{X} \wedge \mathbf{Y}| \mathbf{Z}) \equiv \mathbb{P}(\mathbf{X}|\mathbf{Z})\mathbb{P}(\mathbf{Y}|\mathbf{Z})$ represents a conjunction over all instances of this schema with all combinations of values $\mathbf{X} =\mathbf{x}$, $\mathbf{Y} =\mathbf{y}$, $\mathbf{Z} =\mathbf{z}$. Here $\mathbf{X}, \mathbf{Y}, \mathbf{Z}$ are lists of variables and $\mathbf{x}, \mathbf{y}, \mathbf{z}$ are corresponding instantiations.}
\begin{proposition} \label{dep} $\mathcal{G} \models \mathbb{P}(\mathbf{X} \wedge \mathbf{Y}| \mathbf{Z}) \equiv \mathbb{P}(\mathbf{X}|\mathbf{Z})\mathbb{P}(\mathbf{Y}|\mathbf{Z})$ if and only if $\big((\mathbf{X} \Perp \mathbf{Y})|\mathbf{Z}\big)_{\mathcal{G}}$. In other words, the graph property of d-separation is definable in $\Ll_1$. \qed
\end{proposition} One of the most intriguing components of structural causal reasoning is the \emph{$do$-calculus} \cite{Pearl1995,Pearl2009,Zhang08}, allowing the derivation of causal effects from observational data. This calculus can also be seen as involving graph validity. The next proposition is a slight extension of what was already proved in \cite{Pearl1995}; see also \cite{BCII}. %(The left-to-right direction in each case, which we leave as an exercise, is similar to the left-to-right direction of Prop. \ref{dep}.)
\begin{proposition}\label{docalc} Let $\mathcal{G}$ be a % directed graph
dag over variables $\mathbf{V}$.\footnote{$\mathcal{G}_{\overline{\mathbf{X}}\underline{\mathbf{Z}}}$ is $\mathcal{G}$ minus any edges into $\mathbf{X}$ or out of $\mathbf{Z}$, and $\mathbf{Z}(\mathbf{W})$ is the set of all $\mathbf{Z}$-nodes that are not ancestors of any $\mathbf{W}$-node in $\mathcal{G}_{\overline{\mathbf{X}}}$.}
Then
\begin{enumerate}
 \item $\mathcal{G} \models \mathbb{P}\big([\mathbf{X}]\mathbf{Y}|[\mathbf{X}](\mathbf{Z} \wedge \mathbf{W})\big)  \equiv   \mathbb{P}([\mathbf{X}]\mathbf{Y}|[\mathbf{X}]\mathbf{W})$ iff $\big((\mathbf{Y} \Perp \mathbf{Z}) |\mathbf{X},\mathbf{W}\big)_{\mathcal{G}_{\overline{\mathbf{X}}}}$;
 \item $\mathcal{G} \models \mathbb{P}([\mathbf{X} \wedge \mathbf{Z}]\mathbf{Y}|[\mathbf{X} \wedge \mathbf{Z}]\mathbf{W})  \equiv   \mathbb{P}\big([\mathbf{X}]\mathbf{Y}|[\mathbf{X}](\mathbf{Z} \wedge \mathbf{W})\big)$ iff $\big((\mathbf{Y} \Perp \mathbf{Z}) |\mathbf{X},\mathbf{W}\big)_{\mathcal{G}_{\overline{\mathbf{X}}\underline{\mathbf{Z}}}}$;
 \item $\mathcal{G} \models \mathbb{P}([\mathbf{X} \wedge \mathbf{Z}]\mathbf{Y}|[\mathbf{X}\wedge \mathbf{Z}]\mathbf{W})  \equiv  \mathbb{P}([\mathbf{X}]\mathbf{Y}|[\mathbf{X}]\mathbf{W})$ iff $\big((\mathbf{Y} \Perp \mathbf{Z}) |\mathbf{X},\mathbf{W}\big)_{\mathcal{G}_{\overline{\mathbf{X}} ,\overline{\mathbf{Z}(\mathbf{W})}}}$.
\end{enumerate}
All formulas here are in $\Ll_2$.
\qed
\end{proposition}
We leave further exploration of questions about graph definability in these languages for a future occasion.
% 
% \subsection{Causal Influence}
% Returning to the notion of direct causal influence,
% we now show that $\Ll_3$ always bears witness to such influences.
% Abbreviating, e.g., $X = x$ as $x$:
% \begin{proposition}\label{prop:influence_l3}
% Let $\mathfrak{M} \in \mathcal{M}$. Then the following are equivalent:
% \begin{enumerate}
%  \item $X {\influences_{\mathfrak{M}}} Y$;
%  \item there are $\alpha\in \Lint$, $x, x' \in \Val(X)$, and $y \in \Val(Y)$ such that $\big\llbracket\mathbb{P}\big([\alpha \land x] y \land [\alpha \land x'] \lnot y\big)\big\rrbracket_{\mathfrak{M}} > 0$;
%  \item $\semantics{\mathbb{P}\big([\alpha] y \land [\alpha \land x] \lnot y\big)}_{\mathfrak{M}} > 0$ for some $\alpha, x, y$. \qed
% \end{enumerate}
% \end{proposition}
% \begin{proof}
% $(1) \Rightarrow (2)$:
% \tb{DHI: this is only true for $\mathcal{M}^*$! Consider the model where $\mathbf{V} = \{V_0, V_1, \dots\}$ with $f_{V_1}(\mathbf{v}) = f_{V_2}(\mathbf{v}) = \dots = 0$, and $f_{V_0}(\mathbf{v}) = 1$ if $\mathbf{v}(V_1) = \mathbf{v}(V_2) = \dots = 1$, $f_{V_0}(\mathbf{v}) = 0$ otherwise. Then $V_i \influences V_0$ for every $i \ge 1$ but this is obviously not finitely witnessed. I guess this is good for $\mathcal{M}^*$.}
% \end{proof}
% 
% 
% Finally, note that if $X \influences_{\mathfrak{M}} Y$ then $X \rightarrow Y \in \mathcal{G}_{\mathfrak{M}}$, but not \emph{vice versa}; this observation generalizes the \emph{faithfulness} distinction from Bayesian networks.
% But like Bayesian networks \cite{meekfaithfulness}, SCMs are almost-always ``faithful.'' \tb{cite BCII?}
\section{Axiomatizations}
We now give systems $\AX_i$
each of which axiomatizes the validities of $\Ll_i$ over both $\mathcal{M}$ and $\mathcal{M}^*$.
{Note that since $\Ll_1 \subset \Ll_2 \subset \Ll_3$, the full system
$\AX_3$ completely axiomatizes all three languages.
Its axioms lie in $\Ll_3$, however, while our $\AX_2$ and $\AX_1$ include only principles expressible in the respective lower-level languages.}

These probabilistic logics build on the base (deterministic) logics, as any equivalent base formulas must be assigned the same probability.
We call $\epsilon \in \Lfull = \Lbasethree$ an \emph{$\Lfull$ validity}, and write $\models \epsilon$, if for all $\mathcal{F}$ and $\mathbf{u}$ we have $\mathcal{F},\mathbf{u} \models \epsilon$.
Equivalently, $\models \epsilon$ if $\mathcal{F} \models \epsilon$ for all $\mathcal{F}$ that are deterministic, i.e., lack exogenous dependence.
% Note that for $\epsilon = [\alpha] \beta \in \Ll^{\base}_2 = \Lcond$,
% we have $\models \epsilon$ just in case $\alpha \rightarrow \beta$ is a propositional tautology, and for $\epsilon \in \Lprop$, we have $\models \epsilon$ iff $\epsilon$ is tautologous.

The $\Lfull$ validities for Boolean ranges
have been axiomatized in \cite{II2019}, and the corresponding satisfiability problem is $\mathsf{NP}$-complete.
We need to add one more axiom schema,\footnote{In \cite{II2019}, $\textsf{Def}$ merely amounts to the law of excluded middle since $\Val(X) = \{0, 1\}$ for all $X \in \mathbf{V}$; here we assume only that every $\Val(X)$ is finite.} for every $V \in \mathbf{V}$; here the dummy variables $v, v'$ range over $\Val(V)$:
\begin{eqnarray*}
  \textsf{Def}. && \bigwedge_{\substack{v \neq v'}} \neg [\alpha] \left(V=v \wedge V=v'\right) \wedge \bigvee_{v} [\alpha] (V=v).
\end{eqnarray*}
The first part of \textsf{Def} reflects the uniqueness, and the second part the existence, of solutions under any causal intervention, as ensured by recursiveness (Def.~\ref{def:recursivescm}).

% \hspace{-0.013in}We give the systems in decreasing order of complexity.
\subsection{$\AX_3$}
The inference rule and first $4$ axioms capture propositional and probabilistic reasoning:
\begin{eqnarray*}
  \textsf{MP}. && \text{Inference rule: } \varphi, \varphi \rightarrow \psi \vdash \psi \\
  \textsf{Bool}. && \text{Boolean tautologies over } \Ll_3 \\
  \textsf{NonNeg}. && \mathbb{P}(\epsilon) \geqslant \underline{0} \\
  \textsf{Add}. && \mathbb{P}(\epsilon \land \zeta) + \mathbb{P}(\epsilon \land \lnot \zeta) \equiv \mathbb{P}(\epsilon) \\
  \textsf{Dist}. && \mathbb{P}(\epsilon) \equiv \mathbb{P}(\zeta) \mbox{ whenever } \models \epsilon \leftrightarrow \zeta.
\end{eqnarray*}

A single schema, $\textsf{ProbRec}$, encompasses the interaction between the probability and conditional modalities. It captures the acyclicity of (indirect) causal influence in recursive SCMs (cf. Prop.~\ref{prop:semimarkovisrecursive}).
In $\textsf{ProbRec}$ instances,
$X_1 \neq X_n$ and $x_1 \neq x'_1, \dots, x_n \neq x'_n$.
Abbreviating, e.g., $X_i = x_i$ as $x_i$:
\begin{eqnarray*}
   \textsf{ProbRec}. && 
 \bigwedge_{i=1}^{n-1}\mathbb{P}\big([ \alpha_i \land x_i ] x^*_{i+1} \land [ \alpha_i \wedge  x_i' ] \lnot x^*_{i+1}\big) \not\equiv \0 \\
&&
  \rightarrow
  \mathbb{P}\big(
  [ \alpha_{n} \land x_n ] x^*_1 \land [ \alpha_{n} \wedge x'_n ] \lnot x^*_1
  \big) \equiv \0.
\end{eqnarray*}

Finally, the following 16 axioms, collectively called $\textsf{Poly}$,
constitute a polynomial calculus:
\begin{eqnarray*}
  \textsf{OrdTot}. && \mathbf{t}_1 \geqslant \mathbf{t}_2 \lor \mathbf{t}_2 \geqslant \mathbf{t}_1 \\
  \textsf{OrdTrans}. && \mathbf{t}_1 \geqslant \mathbf{t}_2 \land \mathbf{t}_2 \geqslant \mathbf{t}_3 \rightarrow \mathbf{t}_1 \geqslant \mathbf{t}_3 \\
  \textsf{NonDegen}. && \lnot\left(\underline{0} \equiv \underline{1}\right) \\
  \textsf{AddComm}. && \mathbf{t}_1 + \mathbf{t}_2 \equiv \mathbf{t}_2 + \mathbf{t}_1 \\
  \textsf{AddAssoc}. && (\mathbf{t}_1 + \mathbf{t}_2) + \mathbf{t}_3 \equiv \mathbf{t}_1 + (\mathbf{t}_2 + \mathbf{t}_3) \\ % C2
  \textsf{Zero}. && \mathbf{t} + \underline{0} \equiv \mathbf{t} \\
  \textsf{AddOrd}. && \mathbf{t}_1 \geqslant \mathbf{t}_2 \rightarrow \mathbf{t}_1 + \mathbf{t}_3 \geqslant \mathbf{t}_2 + \mathbf{t}_3 \\
  \textsf{MulOrdG}. && \mathbf{t}_1 \geqslant \mathbf{t}_2 \land \mathbf{t}_3 \geqslant \0 \rightarrow \mathbf{t}_1 \cdot \mathbf{t}_3 \geqslant \mathbf{t}_2 \cdot \mathbf{t}_3 \\
  \textsf{MulOrdL}. && \mathbf{t}_1 \geqslant \mathbf{t}_2 \land \mathbf{t}_3 \leqslant \0 \rightarrow \mathbf{t}_1 \cdot \mathbf{t}_3 \leqslant \mathbf{t}_2 \cdot \mathbf{t}_3 \\
  \textsf{MulComm}. && \mathbf{t}_1 \cdot \mathbf{t}_2 \equiv \mathbf{t}_2 \cdot \mathbf{t}_1 \\
  \textsf{MulAssoc}. && (\mathbf{t}_1 \cdot \mathbf{t}_2) \cdot \mathbf{t}_3 \equiv \mathbf{t}_1 \cdot (\mathbf{t}_2 \cdot \mathbf{t}_3) \\
  \textsf{One}. && \mathbf{t} \cdot \underline{1} \equiv \mathbf{t} \\
  \textsf{MulDist}. && \mathbf{t}_1 \cdot (\mathbf{t}_2 + \mathbf{t}_3) \equiv \mathbf{t}_1 \cdot \mathbf{t}_2 + \mathbf{t}_1 \cdot \mathbf{t}_3 \\
%   \textsf{MulNonNeg}. && \tT_1 \geqslant \0 \land \tT_2 \geqslant \0 \rightarrow \tT_1 \cdot \tT_2 \geqslant \0 \\
  \textsf{ZeroMul}. && \tT \cdot \0 \equiv \0 \\
  \textsf{NoZeroDiv}. && \tT_1 \cdot \tT_2 \equiv \0 \rightarrow \tT_1 \equiv \0 \lor \tT_2 \equiv \0 \\
  \textsf{Neg}. && \mathbf{t} + (- \mathbf{t}) \equiv \underline{0}.
\end{eqnarray*}

\subsection{$\AX_2$}
As for $\AX_2$, note that
$\textsf{Add}$ and $\textsf{ProbRec}$ do not belong to $\Ll_2$ since
$\Lbase_2$ does not close under conjunction.
% If we suppose instead that $\varphi, \psi \in \Lbase_1 = \Lprop$ then
% $\textsf{Add}$ is too weak---it only captures consequents of the trivial antecedent ($\top$) since a purely propositional $\varphi$ in the base language is interpreted as $[\top] \varphi$.
To obtain $\AX_2$, form the axiomatization as above, but replace $\textsf{Add}$ and $\textsf{ProbRec}$ with the variants below.
There is one instance of $\textsf{ProbRec2}$ for every \emph{finite} $\mathbf{W} \subset \mathbf{V}$.
$\mathbf{Y} \prec \mathbf{Z}$ means that $Y \prec Z$ for every $Y \in \mathbf{Y}, Z \in \mathbf{Z}$.
% Also, dummy indices like $\mathbf{p}$ run over instantiations $\Val(\mathbf{P})$, and can be thought of as belonging to $\Lint$.\footnote{$\Lint$ in its entirety corresponds to $\bigcup_{\mathbf{S} \text{ finite}}\Val(\mathbf{S})$.} Likewise, $\mathbf{p} \cup \mathbf{r} \in \Lint$ instantiates $\mathbf{P} \cup \mathbf{R}$.
\begin{eqnarray*}
\textsf{Add2}. && \hspace{-0.2cm}\mathbb{P}\big([\alpha](\beta \land \gamma)\big) + \mathbb{P}\big([\alpha](\beta \land \lnot \gamma)\big) \equiv \mathbb{P}\big([\alpha]\beta\big) \\
\textsf{ProbRec2}. && \bigvee_{\substack{\prec \text{ order} \\ \text{on } \mathbf{W}}} \; \bigwedge_{\substack{\mathbf{X}, \mathbf{Y}, \mathbf{Z} \subseteq \mathbf{W} \\ \mathbf{X} \cap \mathbf{Y} = \varnothing \\ \mathbf{Y} \prec \mathbf{Z} \\ \mathbf{x}, \mathbf{y}, \mathbf{z}}} \mathbb{P}\big([\mathbf{x}]\mathbf{y}\big) \equiv \mathbb{P}\big([\mathbf{x} \land \mathbf{z}]\mathbf{y}\big).
\end{eqnarray*}

In exceptional cases, a probability at the third causal level reduces to the second level.
We need the following principle to capture their nonnegativity:
% Below, $\mathbf{s}$ runs over instantiations of the arbitrary finite variable set $\mathbf{S}$.
\begin{eqnarray*}
\textsf{IncExc}. && \hspace{-0.56cm}\bigwedge_{\substack{\mathbf{Y} \subseteq \mathbf{W}}} \, \sum_{\mathbf{X} \subseteq \mathbf{W} \setminus \mathbf{Y}} (-1)^{\left|\mathbf{X}\right|} \mathbb{P}\left(\big[\mathbf{W} \setminus (\mathbf{X} \cup \mathbf{Y})\big] \mathbf{W}\right) \geqslant \0.
% \textsf{Dist}_2. && \hspace{-0.2cm}\mathbb{P}\big([\alpha] (\beta \land \gamma) \big) \\
% && \qquad\geqslant \mathbb{P}([\alpha\land\beta]\gamma) + \mathbb{P}([\alpha\land\gamma]\beta) - 1 \\
% \textsf{Dist}_2. && \hspace{-0.2cm}\mathbb{P}(\epsilon) \leqslant \mathbb{P}(\zeta) \mbox{ whenever } \models \epsilon \rightarrow \zeta \\
% \textsf{ProbRec}_2. && 
%  \hspace{-0.2cm}\bigwedge_{i=1}^{k-1}\mathbb{P}\big([ \alpha_i \land x_i ] x^*_{i+1}\big) \not\equiv \mathbb{P}\big([\alpha_i \wedge x_i' ] x^*_{i+1}\big) \\
% &&
%   \rightarrow%\bigwedge_{x^*_1}
%   \mathbb{P}\big(
%   [ \alpha_{k} \land x_k ] x^*_1\big) \equiv\mathbb{P}\big([ \alpha_{k} \wedge x'_k ] x^*_1
%   \big).
\end{eqnarray*}

% \tb{$\textsf{IncExc}$ and $\textsf{ProbRec2}$ correspond respectively to the ``inclusion-exclusion inequalities'' and ``directionality'' principle of \cite{TianKP06}.}

\subsection{$\AX_1$}
Conditionals in $\Ll_1$ are trivial, so $\textsf{ProbRec}$ becomes irrelevant.
Consequently $\AX_1$ consists simply of the $5$ schemata $\textsf{MP}$, $\textsf{Bool}$, $\textsf{NonNeg}$, $\textsf{Add}$, $\textsf{Dist}$ with metavariables confined to $\Lbase_1 = \Lprop$, along with $\textsf{Poly}$.

\subsection{Sample Derivation}
Before proving completeness (Thm. \ref{thm:soundcomplete}) we illustrate the power of $\AX_3$ through a representative derivation. Our goal is to derive the example in \cite[\S3.2]{Pearl1995}: \begin{eqnarray} \mathbb{P}\big([x^*]y^*\big) & \equiv & \sum_{z} \mathbb{P}(z|x^*) \sum_{x}\mathbb{P}(y^*|x\wedge z)\mathbb{P}(x) \label{do}\end{eqnarray} 
This formula (in $\Ll_2$) is not in general valid. But it does follow from further assumptions easily statable in $\Ll_2$. Formulas (\ref{second})--(\ref{secondsecond}) below are instances of the second $do$-calculus schema, while (\ref{third}) and (\ref{fourth}) are instances of the third schema. \begin{eqnarray}
  \mathbb{P}\big([X]Z\big) & \equiv & \mathbb{P}(Z|X)\label{second} \\
   \mathbb{P}\big([X]Y|[X]Z\big)   &\equiv&  \mathbb{P}\big([X\wedge Z]Y\big) \label{first}\\
  \mathbb{P}\big([Z]Y|[Z]X\big) & \equiv & \mathbb{P}(Y|X \wedge Z) \label{secondsecond} \\
  \mathbb{P}\big([X \wedge Z]Y\big) & \equiv & \mathbb{P}\big([Z]Y\big) \label{third}\\
  \mathbb{P}\big([Z]X\big) & \equiv &  \mathbb{P}(X)\label{fourth}
\end{eqnarray}
Prop. \ref{docalc} provides the graphical assumptions needed to justify each of these assertions. For example, they are all valid over the graph in Fig. \ref{bow} \cite{Pearl1995,Pearl2009}.
\begin{figure}
\begin{center}
\begin{tikzpicture}
  \node (s0) at (0,0) {$X$};

  \node (s2) at (1,0) {$Z$};

  \node (s3) at (2,0)  {$Y$};

  \node (s4) at (1,.75) {$W$};

  \path (s0) edge[thick,->] (s2);

  \path (s2) edge[thick,->] (s3);

    \path (s4) edge[thick,->] (s0);

  \path (s4) edge[thick,->] (s3);

%  \node (e) at (2.25,-.1) {$=\sum_{C_i}V_{C_i}$};

 \end{tikzpicture} \caption{A graph over which (\ref{second})-(\ref{fourth}) are all valid.}\label{bow} \end{center}
 \end{figure}
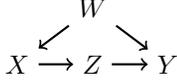
We now argue that $\big((\ref{second})\wedge(\ref{first}) \wedge (\ref{secondsecond}) \wedge (\ref{third}) \wedge (\ref{fourth})\big) \rightarrow (\ref{do})$ is derivable in our calculus.\footnote{It is worth observing that this derivation would go through even if we considered the weaker logic for $\Ll_{\textnormal{full}}$ studied in \cite{II2018}. That is, deriving $\big((\ref{second})\wedge(\ref{first}) \wedge (\ref{secondsecond}) \wedge (\ref{third}) \wedge (\ref{fourth})\big) \rightarrow (\ref{do})$ does not depend on any of the causal axioms that characterize structural causal models \cite{Halpern2000,Pearl2009}. However, slightly weaker assumptions---e.g., $\mathbb{P}([X]Z) \equiv \mathbb{P}([X]Z|X)$ in place of (\ref{second})---would require the additional axioms.}
For simplicity we derive this from $\AX_3$ in $\Ll_3$. 
First, by appeal to \textsf{MP}, \textsf{Bool}, and \textsf{Dist}, we have that $\mathbb{P}([x^*]y^*) \equiv \mathbb{P}\big(\bigvee_{z}([x^*]y^* \wedge [x^*]z)\big)$, which in turn using $\textsf{Add}$ is equal to $\sum_{z}\mathbb{P}([x^*]y^* \wedge [x^*]z)$. By $\Poly$ this can be shown equal to $\sum_{z}\mathbb{P}([x^*]z)\mathbb{P}([x^*]y^* | [x^*]z)$. By  (\ref{second}) and (\ref{first}) this is equal to $\sum_{z}\mathbb{P}(z|x^*)\mathbb{P}([x^*\wedge z]y^*)$, and by (\ref{third}) to $\sum_{z}\mathbb{P}(z|x^*)\mathbb{P}([z]y^*)$. Employing a similar argument to that above (using \textsf{MP}, \textsf{Bool}, $\textsf{Add}$, and \textsf{Dist}), this is equal to $\sum_{z}\mathbb{P}(z|x^*)\sum_{x}\mathbb{P}([z]y^*|[z]x)\mathbb{P}([z]x)$. By (\ref{secondsecond}) and (\ref{fourth}) we finally obtain $\sum_{z} \mathbb{P}(z|x^*) \sum_{x}\mathbb{P}(y^*|x\wedge z)\mathbb{P}(x)$.

% \tb{TODO: make notation across doc uniform, check how this is derivable in $\AX_2$.}

\subsection{Completeness Theorems}
\begin{theorem}
\label{thm:soundcomplete}
Each $\AX_i$ is sound and complete for the validities of $\Ll_i$ with respect to both $\mathcal{M}^*$ and $\mathcal{M}$.
\end{theorem}
\begin{proof}
We first prove the result for $\AX_3$. An analogous argument shows $\AX_1$-completeness.
Then we do $\AX_2$.
\subsubsection{$\AX_3$, $\AX_1$}
Soundness is straightforward.
Completeness: show any consistent $\varphi \in \Ll_3$ is satisfiable.
We work toward a normal form (Lem.~\ref{lem:normalform}).
Let $\mathbf{V}_\varphi \subset \mathbf{V}$ be the finite set of variables appearing in $\varphi$,
let $\mathcal{L}_{\mathrm{full}(\varphi)}$ be the fragment of $\Lfull$ in which only variables from $\mathbf{V}_\varphi$ appear,
and let $\mathcal{L}_{\mathrm{int}(\varphi)} = \Lint \cap \Lfullphi$; the latter language is finite. %= \tb{\bigcup_{\mathbf{V}' \subseteq \mathbf{V}_\varphi} \Val\big(\mathbf{V}'\big)}
Let $\Delta = \big\{ \bigwedge_{\alpha \in \Lintphi} [\alpha] \mathbf{v}_\varphi^\alpha \;\ssep\; \mathbf{v}_\varphi^\alpha \in \Val(\Vphi) \text{ for each } \alpha \big\}$.\footnote{Or, $\Delta = \big\{\bigwedge_{\alpha}[\alpha] f(\alpha) : f \in \Lintphi \to \Val(\Vphi) \big\}$.}
Distinct elements of $\Delta$ are jointly $\Lfull$-unsatisfiable by $\textsf{Def}$.
Let $\Delta_{\mathrm{sat}} = \{ \delta \in \Delta \ssep \centernot{\models}\delta \to \bot\}$.
A preliminary result (where $\vdash_3$ is $\AX_3$ proof):
\begin{lemma}\label{lem:sumform}
Let $\epsilon \in \Lfullphi$. Then
$\vdash_3 \mathbb{P}(\epsilon) \equiv \sum_{\substack{\delta \in \Delta \\ \delta \models \epsilon}} \mathbb{P}(\delta)$.
\end{lemma}
\begin{proof}
We show that $\models \epsilon \leftrightarrow \bigvee_{\delta \models \epsilon} \delta$.
Since $\models \lnot[\alpha]\beta \leftrightarrow [\alpha]\lnot \beta$ and $\models [\alpha](\beta\land\gamma) \leftrightarrow [\alpha]\beta \land [\alpha]\gamma$, we have $\models \epsilon \leftrightarrow [\alpha_1]\beta_1 \land \dots \land [\alpha_n]\beta_n$ with distinct $\alpha_i$'s; since we can conjoin any $[\alpha]\top$, we can suppose every $\alpha \in \Lintphi$ appears.
By $\textsf{Def}$, $\models [\alpha]\beta \leftrightarrow [\alpha]\bigvee_{\mathbf{v}_\varphi \models \beta} \mathbf{v}_\varphi$\footnote{$\mathbf{v}_\varphi \models \beta$, for $\mathbf{v}_\varphi \in \Val(\Vphi)$, means $\models [\top](\mathbf{v}_\varphi \rightarrow \beta)$.} so
% TI: This should not need a separate verification, just following from a familiar normal form result, together with the fact that we can replace equivalents in the consequent of a conditional.}
% Since $\models [\alpha]\lnot(\beta\land \gamma) \leftrightarrow [\alpha](\lnot \beta \lor \lnot \gamma)$ we assume without loss that negations in $\beta_i$ appear only before atoms.
% We prove the claim by induction on $\beta_i$. The base case $\beta_i \in \{x, \lnot x\}$ follows from $\textsf{Def}$.
% Disjunction is trivial.
% Conjunction: $\models \beta \land \gamma \leftrightarrow \big(\bigvee_{\mathbf{v}_\varphi \models \beta} \mathbf{v}_\varphi\big) \land \big(\bigvee_{\mathbf{v}_\varphi \models \gamma} \mathbf{v}_\varphi\big) \leftrightarrow \bigvee_{\substack{\mathbf{v}_\varphi \models \beta \\ \mathbf{v}'_\varphi \models \gamma}} (\mathbf{v}_\varphi \land \mathbf{v}'_\varphi)$. By $\textsf{Def}$ delete any where $\mathbf{v}_\varphi \neq \mathbf{v}'_\varphi$. Thus $\models \beta \land \gamma \leftrightarrow \bigvee_{\substack{\mathbf{v}_\varphi \models \beta \\ \mathbf{v}_\varphi \models \gamma}} (\mathbf{v}_\varphi \land \mathbf{v}_\varphi) \leftrightarrow \bigvee_{\mathbf{v}_\varphi \models \beta \land \gamma} \mathbf{v}_\varphi$
% and 
\begin{align}\label{eqn:epsdisjunction}
% \models \epsilon &\leftrightarrow [\alpha_1]\bigvee_{\mathbf{v}_\varphi \models \beta_1} \mathbf{v}_\varphi \land \dots \land
% [\alpha_n]\bigvee_{\mathbf{v}_\varphi \models \beta_n} \mathbf{v}_\varphi \nonumber\\
% &\leftrightarrow \bigvee_{\substack{\mathbf{v}^1_\varphi \models \beta_1 \\ \dots \\ \mathbf{v}^n_\varphi \models \beta_n}} \big([\alpha_1]\mathbf{v}^1_\varphi \land \dots \land [\alpha_n] \mathbf{v}^n_\varphi\big).
\models \epsilon \leftrightarrow \bigwedge_{i} [\alpha_i]\bigvee_{\mathbf{v}_\varphi \models \beta_i}\mathbf{v}_\varphi
\leftrightarrow \bigvee_{\substack{\mathbf{v}^1_\varphi \models \beta_1 \\ \dots \\ \mathbf{v}^n_\varphi \models \beta_n}} \bigwedge_i [\alpha_i] \mathbf{v}_\varphi^i.
\end{align}

\eqref{eqn:epsdisjunction} disjoins elements of $\Delta$, and any $\delta \notin \Delta_{\cons}$ may be freely appended as a disjunct.
For any $\delta$ appearing in \eqref{eqn:epsdisjunction}, clearly $\delta \models \epsilon$.
Conversely, suppose $\delta \models \epsilon$ but $\delta = \bigwedge_i [\alpha_i] \mathbf{v}_\varphi^i$ does not appear in \eqref{eqn:epsdisjunction}.
Then $\delta \in \Delta_{\cons}$ and $\mathbf{v}_\varphi^{i} \notmodels \beta_i$ for some $i$.
Thus there are $\mathcal{F}$ such that $\mathcal{F} \models [\alpha_i] \mathbf{v}_\varphi^i$ but $\mathcal{F}\notmodels [\alpha_i] \beta_i$, so $\mathcal{F}\notmodels \epsilon$, a contradiction.
Finally, sum over mutually exclusive events.% and follows by $\textsf{Add}$, $\textsf{Dist}$, $\textsf{Def}$.
\end{proof}
There is a consistent clause in the disjunctive normal form of $\varphi$ so we assume that $\varphi \in \Ll_3$ is a conjunction of literals (using $\textsf{MP}$, $\textsf{Bool}$) and apply the following:
\begin{lemma}
  \label{lem:normalform}
Let $\varphi$ be a conj. of lit.
There are polynomial terms %$\{\mathbf{t}_i, \mathbf{t}'_{i'}\}_{\substack{1 \le i \le m \\ 1 \le i' \le m'}}$
$\{\mathbf{t}_i, \mathbf{t}'_{i'}\}_{\substack{i, i'}}$
in the variables $\big\{\mathbb{P}(\delta)\big\}_{\delta \in \Delta}$
such that % $\varphi$ is provably-in-$\AX_3$ equivalent to
  % \begin{equation}
    % \label{eqn:normalform}
    % \begin{gathered}
      \begin{multline}\label{eqn:normalform}
%     \bigvee_{\Delta' \subseteq \Delta} \Bigg[
  \vdash_3 \varphi\leftrightarrow\bigvee_{\Delta' \subseteq \Delta}\Bigg[
    \bigwedge_{\delta \in \Delta'} \mathbb{P}(\delta) > \0 \land \bigwedge_{\delta \notin \Delta'} \mathbb{P}(\delta) \equiv \0 \\
    \land \sum_{\delta \in \Delta} \mathbb{P}\left(\delta\right) \equiv \underline{1}
%       \bigwedge_{\delta \in \Delta_{\bot}}
%         \mathbb{P}(\delta) \equiv \underline{0}
%       \;\land
%       \bigwedge_{\delta \in \Delta \setminus \Delta_{\bot}}
%         \mathbb{P}(\delta) \geqslant \underline{0}
      \land \bigwedge_{i}
        \mathbf{t}_i \geqslant \underline{0}
       \land
      \bigwedge_{i'}
        \mathbf{t}'_{i'} > \underline{0} \Bigg].
      \end{multline}
    % \end{gathered}
  % \end{equation}
\end{lemma}
% Let $\Delta_{\bot} = \{ \delta \in \Delta : \; \models \delta \leftrightarrow \bot \}$.
\begin{proof}
% We will show the existence of $\{\mathbf{t}_i, \mathbf{t}'_{i'}\}_{i, i'}$ such that
% $\varphi \leftrightarrow \bigwedge_{i=1}^m
% \mathbf{t}_i \geqslant \underline{0}
% \land
% \bigwedge_{i=1}^{m'}
% \mathbf{t}'_i > \underline{0}$.
The disjunction in \eqref{eqn:normalform} and the first three conjuncts within each disjunct follow easily by propositional reasoning, $\textsf{Dist}$ (including $\textsf{Def}$), $\textsf{NonNeg}$, and $\textsf{Add}$.
To obtain the last two conjuncts, since $\varphi$ is a conjunction of $\Ll_3$ literals 
suppose $\varphi = \bigwedge_{i} \mathbf{f}_i \geqslant \0 \land \bigwedge_{i'} \lnot(\mathbf{f}'_{i'} \geqslant \0)$.
By $\textsf{OrdTot}$, $\textsf{AddOrd}$, $\textsf{OrdTrans}$, $\textsf{AddComm}$, $\textsf{Neg}$, $\textsf{Zero}$,
$\varphi \leftrightarrow \bigwedge_{i} \mathbf{f}_i \geqslant \0 \land \bigwedge_{i'} (-\mathbf{f}'_{i'} > \0)$.
We claim each $\mathbf{f}_i, -\mathbf{f}'_{i'}$ has an equivalent polynomial in the variables $\big\{\mathbb{P}(\delta)\big\}_{\delta \in \Delta}$. This gives the $\{\mathbf{t}_i, \mathbf{t}'_{i'}\}_{i, i'}$.
It can be easily shown by induction on polynomial terms
that $\textsf{Poly}$ allows replacement of equivalents.\footnote{That is, that if $\vdash \mathbf{t}_1 \equiv \mathbf{t}_2$ then provably $\mathbf{t}_1$ can be replaced with $\mathbf{t}_2$ in any composite term in which it appears. Both $\textsf{MulOrd}$ directions are crucial to the proof.}
Replacing equivalents under Lem.~\ref{lem:sumform} finishes the proof.
% Thus it suffices to show that $\mathbb{P}(\epsilon)$, for arbitrary $\epsilon \in \Lfullphi$, has an equivalent polynomial in the variables $\big\{\mathbb{P}(\delta)\big\}_{\delta \in \Delta}$.
\end{proof}

Some disjunct of (\ref{eqn:normalform}), say that where $\Delta'= \Delta^*$, is consistent so we suppose $\varphi$ is such disjunct.
This is also a polynomial system $S$ over unknowns $\big\{\mathbb{P}(\delta)\big\}_{\delta \in \Delta}$. We now show it has a solution. Our primary tool is the following semialgebraic result \cite{Stengle1974}.
\begin{theorem}[Positivstellensatz]
\label{thm:psatz}
Let $R = \mathbb{Q}[x_1, \dots, x_n]$ and
$F, G, H$ be finite sets of polynomials in $R$. Let $\mathrm{cone}(G) \subseteq R$ be the closure of $G \cup \{s^2 : s \in R \}$ under $+$ and $\times$, and let $\mathrm{ideal}(H) = \big\{\sum_{h \in H} a_h h : a_h \in R \text{ for each } h \big\}$.
Then either $\big\{ f \neq 0, g \ge 0, h = 0 : (f, g, h) \in (F, G, H) \big\}$ has a solution over $\mathbb{R}^n$, or there exist $g \in \mathrm{cone}(G)$, $h \in \mathrm{ideal}(H)$, $n \in \mathbb{N}$
such that
\begin{equation}
  \label{eqn:certinf}
  g + h + f^{2n} = 0
\end{equation}
where $f = \prod_{f' \in F} f'$.
\qed
\end{theorem}
Each clause in $S$ easily translates to a polynomial in  Thm. \ref{thm:psatz}; a clause $\tT'_i > \0$ becomes two constraints: $\tT'_i \ge 0$ and $\tT'_i \neq 0$. If there's no solution, let $\mathbf{t} = (\mathbf{g} + \mathbf{h}) + \mathbf{f}^{2n}$ for some $\mathbf{g}, \mathbf{h}, \mathbf{f}$ as in (\ref{eqn:certinf}), where $\mathbf{f}^{2n}$ is an iterated multiplication.
We claim $\varphi \vdash \mathbf{t} \equiv \underline{0} \land \mathbf{t} > \underline{0}$ so that $\varphi$ is inconsistent, a contradiction.
We use the principles below, all derivable from $\Poly$:
  \begin{eqnarray*}
    \textsf{AddPos}. && \tT_1 \geqslant \0 \land \tT_2 > \0 \rightarrow \tT_1 + \tT_2 > \0 \\
    \textsf{MulPos}. && \tT_1 > \0 \land \tT_2 > \0 \rightarrow \tT_1 \cdot \tT_2 > \0 \\
    \textsf{NegAdd}. && -(\tT_1 + \tT_2) \equiv (-\tT_1) + (-\tT_2) \\
    \textsf{NegMul}. && - (\tT_1 \cdot \tT_2) \equiv (-\tT_1) \cdot \tT_2 \\
    \textsf{NegNeg}. && -(-\tT) \equiv \tT \\
    \textsf{OrdSq}. && \tT \cdot \tT \geqslant \0.
  \end{eqnarray*}
First, we show $\varphi \vdash \tT > \0$. Note that
$\varphi \vdash \mathbf{g} \geqslant \0$ by $\textsf{OrdSq}$ and $\varphi \vdash \mathbf{h} \equiv \0$ by $\textsf{ZeroMul}$ given Thm. \ref{thm:psatz} and $S$. % $\textsf{MulEq}$, $\textsf{Zero}$.
Also, $\varphi \vdash \mathbf{f} > \0$ by $\textsf{NonDegen}$ if $m' = 0$ in $S$ or $n = 0$ in (\ref{eqn:certinf}) and by $\textsf{MulPos}$ otherwise.
So $\varphi \vdash (\mathbf{g} + \mathbf{h}) + \mathbf{f}^{2n} > \0$ by $\textsf{AddPos}$.
Now we show $\varphi \vdash
\mathbf{t} \equiv \underline{0}$. In fact, we don't need $\varphi$. We show that $\Poly$ is powerful enough to simplify polynomials; then by soundness and since (\ref{eqn:certinf}) holds identically,
$\vdash \mathbf{t} \equiv \underline{0}$.
Using $\textsf{MulDist}$, $\textsf{NegAdd}$, $\vdash \mathbf{t} \equiv \mathbf{m}$ where $\mathbf{m}$ is a sum of non-$\0$ monomials, each of which is either $\1$ itself or contains no factors of $\1$ ($\One$), and contains at most one $-$ sign ($\textsf{NegMul}$, $\textsf{NegNeg}$). By $\textsf{MulComm}, \textsf{MulAssoc}$ group the factors in each left-associatively and in increasing (lexicographic) order of their variables.
Then with $\AddComm$, $\AddAssoc$, $\textsf{Neg}$ group and cancel out to $\0$ equal but opposite monomials. Adding all the $\0$s, we have $\vdash \mathbf{m} \equiv \0$.

Thus $S$ has solution $s$. We move toward constructing a model satisfying \eqref{eqn:normalform}.
Define $X \influencesindirect_\delta Y$ if $\models \delta \rightarrow [\alpha \land x] y \land [\alpha \land x'] \lnot y$ for some $\alpha, x \neq x', y$.
\begin{lemma}\label{lem:metamodel}
Let $\mathfrak{P} : \Delta \to [0, 1]$ be computable and suppose
$\sum_{\delta} \mathfrak{P}(\delta) = 1$,
that $\mathrm{supp}(\mathfrak{P}) \subset \Delta_{\mathrm{sat}}$,\footnote{$\supp(\mathfrak{P}) = \{\delta \ssep \mathfrak{P}(\delta) > 0\}$ is the \emph{support} of $\mathfrak{P}$.}
and $\bigcup_{\delta \in \mathrm{supp}(\mathfrak{P})} \influencesindirect_{\delta}$ is acyclic.
Then there is $\mathfrak{M} \in \mathcal{M}^*$ such that $\semantics{\mathbb{P}(\delta)}_{\mathfrak{M}} = \mathfrak{P}(\delta)$ for every $\delta \in \Delta$.
\end{lemma}
\begin{proof}
For each $\delta \in \supp(\mathfrak{P}) \subset \Delta_{\cons}$, % by $\Lfull$ consistency
there are structural mechanisms $\mathcal{F}(\delta) = \big\{f^{\delta}_V\big\}_{V}$
such that %$f^{\delta^*}_V$ do not depend on any $U$, are constant if $V \notin \Vphi$,
for all $V \notin \Vphi$, $f^\delta_V$ is constant; $f^\delta_V(\mathbf{v}_1) = f^\delta_V(\mathbf{v}_2)$ if $\mathbf{v}_2(V_\varphi) = \mathbf{v}_2(V_\varphi)$ for all $V_\varphi \in \Vphi$;
% $\influences_{\mathcal{F}(\delta)} \subset \Vphi^2$
and $\mathcal{F}(\delta) \models \delta$ \cite[Thm.~2]{II2019}.
Consider $\mathfrak{M} = (\mathcal{F}, P)$ with one exogenous variable $\mathcal{U} = \{ U \}$ where $\text{Val}(U) = \supp(\mathfrak{P})$ and $P(U = \delta) = s(\delta)$.
Define $\mathcal{F} = \{f_V\}_{V}$ by $f_{V}(\mathbf{v}, \delta) = f_{V}^{\delta}(\mathbf{v})$.
$\mathcal{F}, \delta_1 \models \delta_2$ iff $\delta_1 = \delta_2$ since elements of $\Delta$ are jointly $\Lfull$-unsatisfiable. %; similarly $\mathcal{F}, \delta \notmodels \delta'$ if $\delta' \notin \supp(\mathfrak{P})$.
Thus %$S_{\mathfrak{M}}(\delta) = \{ \delta \}$ and
$\semantics{\mathbb{P}(\delta)}_{\mathfrak{M}} = s(\delta)$ for all $\delta \in \Delta$. Clearly $\mathfrak{M}$ can be made to satisfy Def. \ref{compmodel}, since $\mathfrak{P}$ is computable.

% We check that $\rightsquigarrow_{\mathfrak{M}}$ is well-founded. % \tb{TI: What is the argument for this?}
We claim $\left.\influences_{\mathfrak{M}}\right. = \bigcup_{\delta \in \supp(\mathfrak{P})} \influences_{\mathcal{F}(\delta)}\footnote{The definition of influence implies that for a deterministic model $\mathcal{F} = \{f_V\}_{V}$, $X \influences_{\mathcal{F}} Y$ if there are $\mathbf{v}_1, \mathbf{v}_2 \in \Val(\mathbf{V})$ such that $\mathbf{v}_1(V) \neq \mathbf{v}_2(V)$ iff $V = X$ and $f_Y(\mathbf{v}_1) \neq f_Y(\mathbf{v}_2)$.} \subseteq \bigcup_{\delta \in \supp(\mathfrak{P})} \influencesindirect_\delta$\footnote{The last is an inclusion rather than an equality: $\influences$ is direct influence while $\influencesindirect$ is indirect.} so that $\influences_{\mathfrak{M}}$ is a dag on $\Vphi$; along with constancy of each $\mathcal{F}(\delta)$ outside $\Vphi$ this shows $\mathfrak{M}$ recursive (cf. Prop.~\ref{prop:semimarkovisrecursive}).
The equality is trivial.
To obtain the inclusion, if $X \influences_{\mathcal{F}(\delta)} Y$ then $X, Y \in \Vphi$, and
there are $\mathbf{v}_1, \mathbf{v}_2$ differing only at $X$ such that $y_1 \neq y_2$ where $y_1 = f_Y^\delta(\mathbf{v}_1)$ and $y_2 = f_Y^\delta(\mathbf{v}_2)$.
Let $\mathbf{Z} = \Vphi \setminus \{X, Y\}$, let
$\mathbf{z} = \mathbf{v}_1(\mathbf{Z}) = \mathbf{v}_2(\mathbf{Z})$, and let $x_1 = \mathbf{v}_1(X) \neq x_2 = \mathbf{v}_2(X)$.
Since $Y \ninfluences_{\mathcal{F}(\delta)} Y$ and $V \ninfluences_{\mathcal{F}(\delta)} Y$ for any $V \notin \Vphi$,
$\mathcal{F}(\delta) \models [\mathbf{z} \land x_1] y_1 \land [\mathbf{z} \land x_2]y_2$.
Let $\delta = [\mathbf{z}\land x_1]\mathbf{v}^{\mathbf{z}\land x_1}_\varphi \land \dots \land [\mathbf{z}\land x_2] \mathbf{v}_\varphi^{\mathbf{z}\land x_2}$.
If $\mathbf{v}^{\mathbf{z}\land x_1}_\varphi(Y) \neq y_1$ or $\mathbf{v}^{\mathbf{z}\land x_2}_\varphi(Y) \neq y_2$ then since $\mathcal{F}(\delta) \models \delta$, we have $\mathcal{F}(\delta) \models \bot$, a contradiction.
Thus $\models \delta \rightarrow [\mathbf{z}\land x_1]y_1 \land [\mathbf{z}\land x_2] \lnot y_1$ and $X \influencesindirect_\delta Y$.
% The $\supseteq$ direction is trivial.
% By (3) of the same, we can suppose that $\alpha \in \Lintphi$. But the resulting witness lies in $\Lfullphi$, and is a truncation of $\delta$, so that $X \rightsquigarrow_\delta Y$.
% If $X \influences_{\mathfrak{M}} Y$ then there is some $\delta \in \supp(\mathfrak{P})$ and $\mathbf{v}_1, \mathbf{v}_2$ differing only at $X$ such that $f_Y(\mathbf{v}_1, \delta) \neq f_Y(\mathbf{v}_2, \delta)$.
% But since $f_Y(\mathbf{v}, \delta) = f^\delta_Y(\mathbf{v})$ for any $\mathbf{v}$ this entails that $X \influences_{\mathcal{F}(\delta)} Y$.
% To see that every $\mathbf{S} \subseteq \mathbf{V}$ has a $\influences_{\mathfrak{M}}$-minimal $M \in S$.
% If not, any $V \notin \Vphi$ is minimal.
\end{proof}
Note $s : \Delta \to [0, 1]$ can be assumed computable, as $S$ has an algebraic \cite{Tarski1949-TARADM-3} and \emph{a fortiori} computable solution.
Since $\mathrm{supp}(s) = \Delta^*$, by $\textsf{Dist}$ and $\textsf{ProbRec}$, applying Lem.~\ref{lem:metamodel} with $\mathfrak{P} = s$ gives $\mathfrak{M}\models \varphi$.
\subsubsection{$\AX_2$}
Soundness: only
$\textsf{IncExc}$ is nontrivial.
We leave it to the reader to generalize
the following $\AX_3$ derivation of a prototypical instance.
By $\textsf{NonNeg}$, we have $\mathbb{P}\big([y \land z] \lnot x \land [x \land z] \lnot y \land [x\land y] z\big) \geqslant \0$.
Marginalizing with $\textsf{Add2}$, we obtain
$\mathbb{P}\big([x\land y] z \big) - \mathbb{P}\big([x\land y] z \land ([y\land z] x \lor [x \land z] y) \big) \geqslant \0$.
By $\textsf{Dist}$ and $\mathbb{P}(\epsilon \lor \zeta) \equiv \mathbb{P}(\epsilon) + \mathbb{P}(\zeta) - \mathbb{P}(\epsilon \land \zeta)$ we have that
$\mathbb{P}\big([x\land y] z\big) - \mathbb{P}\big([x\land y] z \land [y \land z] x\big) - \mathbb{P}\big([x \land y] z \land [x\land z] y \big) + \mathbb{P}\big([x\land y] z \land [y \land z] x \land [x \land z] y\big) \geqslant \0$.
Now $\textsf{Dist}$ribute ``reversibility'' $\models [\alpha](\beta\land\gamma) \leftrightarrow [\alpha\land\beta]\gamma \land [\alpha\land\gamma]\beta$ \cite{GallesPearl,Halpern2000}  across the last three terms.
Obtain at last $\mathbb{P}\big([x\land y] z\big) - \mathbb{P}\big([y] (x \land z)\big) - \mathbb{P}\big([x ](y\land z) \big) + \mathbb{P}(x\land y \land z) \geqslant \0$, the conjunct in $\textsf{IncExc}$ corresponding to $\mathbf{w} = x \land y \land z$, $\mathbf{Z} = \{Z \}$.

Completeness: we obtain analogues of Lem.~\ref{lem:sumform}, \ref{lem:normalform}.
% Let $\mathbf{V}_\varphi = \{X_1, \dots, X_n\} \subset \mathbf{V}$ be the (endogenous) variables appearing anywhere in $\varphi$
% and let $\Delta = \big\{ [\alpha] \mathfrak{v} : \alpha \in \mathcal{L}_{\mathrm{int}(\varphi)}, \mathfrak{v} \in \Val\left(\mathbf{V}_\varphi\right) \big\}$.
% Both are finite.
Let $\textsf{IncExc}_\varphi$ be the instance of $\textsf{IncExc}$ for $\mathbf{W} = \mathbf{V}_\varphi$.
Where $\prec$ is an order on $\mathbf{V}_\varphi$, let
$\textsf{ProbRec2}_{\varphi}(\prec)$ be the $\Ll_2$
\begin{equation*}
\label{eq:probrec2phi}
% \bigvee_{\substack{< \text{ order} \\ \text{on } \mathbf{V}_\varphi}} \;
\bigwedge_{\substack{\mathbf{X}, \mathbf{Y}, \mathbf{Z} \subseteq \mathbf{V}_\varphi \\ \mathbf{X} \cap \mathbf{Y} = \varnothing \\ \mathbf{Y} \prec \mathbf{Z} \\ \mathbf{x}, \mathbf{y}, \mathbf{z}}}
\sum_{\mathbf{v}_\varphi \models \mathbf{y}}%\substack{\mathbf{v}_\varphi \text{ cons.}\\\text{with } \mathbf{q}}}
\mathbb{P}\big([\mathbf{x}]\mathbf{v}_\varphi\big)
\equiv
\sum_{\mathbf{v}_\varphi \models \mathbf{y}}%\substack{\mathbf{v}_\varphi \text{ cons.}\\\text{with } \mathbf{q}}}
\mathbb{P}\big([\mathbf{x} \land \mathbf{z}]\mathbf{v}_\varphi\big).
\end{equation*}
%{Here, that $\mathbf{v}_\varphi$ is consistent with $\mathbf{q}$ means $\mathbf{v}_\varphi(Q) = \mathbf{q}(Q)$ for every $Q \in \mathbf{Q}$, or $\mathbf{v}_\varphi \models \mathbf{q}$.}
% We obtain analogues of Lem.~\ref{lem:sumform}, \ref{lem:normalform}.
Define $\Delta_2 = \big\{[\alpha]\mathbf{v}_\varphi : \alpha \in \Lintphi, \mathbf{v}_\varphi \in \Val(\Vphi)\big\}$
and let $\Lcondphi = \Lcond \cap \Lfullphi$.
% where $\alpha$ runs over $\Lintphi$
Below, $\vdash_2$ will denote $\AX_2$ provability.
\begin{lemma} \label{lem:sumform2}
Let $[\alpha]\beta \in \Lcondphi$. Then
 $\vdash_2 \mathbb{P}\big([\alpha]\beta\big) \equiv \sum_{\mathbf{v}_\varphi \models \beta}%\substack{\mathbf{v}_\varphi \text{ cons.}\\\text{with } \mathbf{q}}}
\mathbb{P}\big([\alpha]\mathbf{v}_\varphi\big)$. \qed
\end{lemma}
\begin{proof}
Like that of Lem.~\ref{lem:sumform}, but using $\textsf{Add2}$ and $\textsf{Def}$.
\end{proof}
\begin{lemma}
\label{lem:normalform2}
Let $\varphi$ be a conj. of lit.
There are polynomial terms %$\{\mathbf{t}_i, \mathbf{t}'_{i'}\}_{\substack{1 \le i \le m \\ 1 \le i' \le m'}}$
$\{\mathbf{t}_i, \mathbf{t}'_{i'}\}_{{i, i'}}$
in the variables $\big\{\mathbb{P}(\delta)\big\}_{\delta \in \Delta_2}$
such that% $\varphi$ is provably-in-$\AX_2$ equivalent to
% \footnote{Concretely these are
% and
% \begin{equation}
% \label{eq:incexcphi}
% \bigwedge_{\substack{\mathcal{R} \subseteq \mathbf{V}_\varphi\\ {\mathbf{v}_\varphi}}} \, \sum_{\mathcal{Q} \subseteq \mathbf{V}_\varphi \setminus \mathcal{R}} (-1)^{\left|\mathcal{Q}\right|} \mathbb{P}\big([\mathbf{v}_\varphi \setminus (\mathfrak{q} \cup \mathfrak{r})] \mathbf{v}_\varphi\big) \geqslant \0.
% \end{equation}
% In both, the dummy index $\mathbf{v}_\varphi$ runs over $\Val\left(\mathbf{V}_\varphi\right)$.}
% \begin{equation}
%   \begin{gathered}
\begin{multline}  \label{eqn:normalform2}
\vdash_2 \varphi \leftrightarrow
\bigvee_{\prec}\Bigg[\bigwedge_{\delta \in \Delta_2 } \mathbb{P}(\delta) \geqslant \0 \land \bigwedge_{\alpha} \sum_{\mathbf{v}_\varphi} \mathbb{P}\big([\alpha]\mathbf{v}_\varphi\big) \equiv \1 \\
    \land \bigwedge_{\alpha}
    \sum_{\mathbf{v}_\varphi \models \alpha}%\substack{\mathbf{v}_\varphi \text{ cons.}\\\text{with }\alpha}}
      \mathbb{P}\big([\alpha] \mathbf{v}_\varphi\big) \equiv \underline{1}
    \land \textsf{ProbRec2}_\varphi(\prec) \\ \land \textsf{IncExc}_\varphi
    \land \bigwedge_{i}
      \mathbf{t}_i \geqslant \underline{0}
    \land
    \bigwedge_{i'}
      \mathbf{t}'_{i'} > \underline{0}\Bigg].
\end{multline}
%   \end{gathered}
% \end{equation}
% is $\AX_2$-consistent and $\models \eqref{eqn:normalform2} \rightarrow \varphi$.
\end{lemma}
\begin{proof}
Like that of Lem.~\ref{lem:normalform}, but using Lem.~\ref{lem:sumform2}. Note that $\vdash_2 \mathbb{P}\big([\alpha]\alpha\big) \equiv \1$ since $\models [\alpha]\alpha$.
% Take note of Lem.~\ref{lem:sumform2}.
% Consider $\textsf{ProbRec2}$ and $\textsf{IncExc}$ for $\mathbf{W} = \mathbf{V}_\varphi$, finding that
% $\vdash_2 \varphi \leftrightarrow \bigvee_< \big[\textsf{ProbRec2}_\varphi(<) \land \textsf{IncExc}_\varphi \land \varphi\big]$.
% Obtain the $\{\mathbf{t}_i, \mathbf{t}'_{i'}\}_{i, i'}$ by applying $\textsf{Add2}$ to each clause of $\varphi$ (see proof of Lem.~\ref{lem:normalform}).
% The remaining conjuncts in \eqref{eqn:normalform2} come from $\textsf{NonNeg}$, $\textsf{Add2}$, and $\textsf{Dist}$ since $\models [\alpha]\alpha$.
\end{proof}
We can apply Lem.~\ref{lem:normalform2} since we can again take $\varphi$ to be a conjunction of $\Ll_2$ literals.
Let $\prec^*$ be one whose disjunct in
(\ref{eqn:normalform2}) is
consistent.
This disjunct is 
a polynomial system $S$.
Being consistent with $\textsf{Poly}$, $S$ has a solution $s$, which constitutes a \emph{$\mathbf{P_*}$ set} for $\mathbf{V}_\varphi$ in the parlance of \cite{TianKP06},
where an exact characterization of such sets is obtained.
$S$ entails all criteria\footnote{Namely, ``effectiveness, recursiveness, directionality, and inclusion-exclusion inequalities.'' Directionality holds for the order $\prec^*$.}
of this characterization,
so there exists $\mathfrak{M}$ inducing $s$ \cite[Thm.~2]{TianKP06}.
Since $s$ can again be taken computable we can take $\mathfrak{M} \in \mathcal{M}^*$.
\end{proof}

\section{Complexity}
Let $\textsc{Prob-Causal-Sat}_i$ be the problem of deciding if a given formula $\varphi \in \mathcal{L}_i$ is satisfiable. The completeness proof above delivers only very loose complexity bounds, which we will now tighten to polynomial space. As $\mathsf{PSPACE}$ closes under complement, this shows validity is also $\mathsf{PSPACE}$.
% input encoding removed---not necessary.
%The following is assumed about the encoding of input.
%Fixing labels $\mathbf{V} = \{ V_1, \dots \}$
%and $\Val(V_n) = \{1, \dots, N_n\}$, we assume that encoding the
%atom $V_n = m \in \Lint \cap \Lprop$ takes $O(\lg n + \lg m)$ bits.

% Toward refining the normal form we introduce the deterministic language $\Ll^+$ from \cite{II2019}, which is like $\Lfull$ but includes propositions over an explicit indirect influence relation $\influencesindirect$: where $X, Y \in \mathbf{V}$,
% \begin{eqnarray*}
% \Ll^+ & ::= & \mathcal{L}_{\textnormal{cond}} \quad | \quad X \influencesindirect Y \quad | \quad \neg \Ll^+ \quad | \quad  \Ll^+ \wedge \Ll^+.
% \end{eqnarray*}
% For a deterministic model $\mathcal{F} = \{f_V\}_V$, we define $\mathcal{F} \models X \influencesindirect Y$ if there is an $\alpha \in \Lint$, $x \neq x' \in \Val(X)$, and $y \neq y' \in \Val(Y)$ such that $\mathcal{F} \models [\alpha\land x] y \land [\alpha \land x'] y'$.
% Probabilistic semantics are defined as before,\footnote{Note that if $\mathfrak{M}$ has recursive order type $\omega$, then $X \influences_{\mathfrak{M}} Y$ entails $\semantics{\mathbb{P}(X \influencesindirect Y)}_{\mathfrak{M}} > 0$.}
% and we assume that encoding $V_i \influencesindirect V_j$ takes $O(\lg i + \lg j)$ bits.

For each $V \in \mathbf{V}_\varphi$, let $\Val_\varphi(V) \subset \Val(V)$ be the subset of values for $V$ appearing explicitly in $\varphi$, and let $\Lintphi^- \subset \Lintphi$ be the subset of interventions appearing in $\varphi$. 
For each $V \in \mathbf{V}_\varphi$, let $B(V) = \{ V = v\}_{v \in \Val_\varphi(V)} \cup \{ \cancel{\beta}_V \} \subset \Lprop$, where $\cancel{\beta}_V = \bigwedge_{v\in \Val_\varphi(V)} V \neq v$.
Then define a subset $\Delta_\varphi \subset \Lfullphi$
by $\Delta_\varphi = \Big\{ \bigwedge_{\substack{\alpha \in \Lintphi^-}} \big([\alpha] \bigwedge_{V \in \mathbf{V}_\varphi} \beta_V^\alpha \big) : \beta_V^\alpha \in B(V) \text{ for each } \alpha \in \Lintphi^-, V \in \Vphi \Big\}$.
Let $E \subset \Lfullphi$ be the set of base formulas appearing in $\varphi$, i.e., those $\epsilon$ such that $\mathbb{P}(\epsilon)$ appears in $\varphi$.
Where $\prec$ is a total order on $\mathbf{V}_\varphi$, let
$\mathcal{M}_\prec$ be the class of deterministic SCMs recursive over $\prec$ and let
$\Delta_\prec = \{\delta \in \Delta_\varphi : \delta \text{ is satisfiable in }\mathcal{M}_\prec\}$.
The next two results are analogous to Lem.~\ref{lem:sumform},~\ref{lem:metamodel} resp.:
\begin{lemma}\label{lem:sumform:small:order}
 If $\epsilon \in E$ and $\prec$ is the restriction of the recursive order of $\mathfrak{M}$ to $\mathbf{V}_\varphi$,
 then $\mathfrak{M} \models \mathbb{P}(\epsilon) \equiv \sum_{\substack{\delta \in \Delta_\prec \\ \delta \models \epsilon}} \mathbb{P}(\delta)$.
\end{lemma}
\begin{proof}
We first show that $\models \mathbb{P}(\epsilon) \equiv \sum_{\substack{\delta \in \Delta_\varphi \\ \delta \models \epsilon}} \mathbb{P}(\delta)$.
This can be shown similarly to Lem.~\ref{lem:sumform}; % and \cite[A.1,~Lem.~2]{BCII}.
the only piece worth verifying is that $\models [\alpha]\beta \leftrightarrow [\alpha]\bigvee_{\{\beta^\alpha_V\}_V \models \beta} \bigwedge_{V} \beta^\alpha_V$ for any $[\alpha]\beta \in \Lcondphi$.
This follows since, by propositional logic we can assume $\lnot$ to occur only before atoms in $\beta$, and
$\models V \neq v \leftrightarrow \cancel{\beta}_V \lor \bigvee_{\substack{v' \in \Val_\varphi(V)\\ v' \neq v}} V = v'$
by $\textsf{Def}$, for any $V \in \mathbf{V}_\varphi$ and $v \in \Val_\varphi(V)$.
To get the final result note that if $\delta \notin \Delta_\prec$ then $\mathfrak{M} \models \mathbb{P}(\delta) \equiv \0$.
\end{proof}
\begin{lemma}\label{lem:metamodel:2}
Suppose $\mathfrak{P} : \Delta_\varphi \to [0, 1]$ is computable and
$\mathrm{supp}(\mathfrak{P}) \subset \Delta_\prec$ for some $\prec$.
Then there is $\mathfrak{M} \in \mathcal{M}^*$ such that $\semantics{\mathbb{P}(\delta)}_{\mathfrak{M}} = \mathfrak{P}(\delta)$ for every $\delta \in \Delta_\varphi$.
\qed
\end{lemma}
The crux is the following small-model property.
\begin{lemma}
  \label{lem:smallmodel}
Any satisfiable $\varphi$ has a \emph{small} model $\mathfrak{M}$ in the sense that $\left|\big\{ \delta \in \Delta_\varphi : \semantics{\mathbb{P}(\delta)}_{\mathfrak{M}} > 0 \big\}\right| \le |\varphi|$.
\end{lemma}
\begin{proof}
We have $|E| < |\varphi|$. 
Let $\mathfrak{M}'$, $\prec$ be such that $\mathfrak{M}' \models \varphi$ and $\mathfrak{M}'$ is recursive with order $\prec$ on $\mathbf{V}_\varphi$.
Consider the system
$S = \big\{\sum_{\delta} \mathbb{P}(\delta) = 1 \big\}
\cup
\big\{ \sum_{\substack{\delta \models \epsilon}} \mathbb{P}(\delta) = {\semantics{\mathbb{P}(\epsilon)}_{\mathfrak{M}'}} \big\}_{\epsilon\in E}$ in the unknowns $\big\{ \mathbb{P}(\delta) \big\}_{\delta \in \Delta_\prec}$, which has a solution by Lem.~\ref{lem:sumform:small:order}.
By a fact of linear algebra \cite[p.~145]{Chv1983}, since $S$ has $|E| + 1$ equations, it has a nonnegative solution $s$ where at most $|E| + 1$ variables are nonzero. Apply Lem.~\ref{lem:metamodel:2} to $s$.
% By Lem.~\ref{lem:sumform},
% {$\semantics{\mathbb{P}(\epsilon)}_{\mathfrak{M}} = \semantics{\mathbb{P}(\epsilon)}_{\mathfrak{M}'}$} for all $\epsilon \in E$,
% so $\mathfrak{M} \models \varphi$.
%\tb{$\influences_{\mathfrak{M}} \subseteq \bigcup_{\delta\in\Delta(\mathfrak{M}')}\influences_\delta = \influences_{\mathfrak{M}'}$} and is hence well-founded.
\end{proof}

\begin{lemma}\label{lem:detlogic:p}
Given $\delta \in \Delta_\varphi$ and $\prec$, deciding whether or not $\delta \in \Delta_\prec$ is $\mathsf{P}$ in time parameter $|\varphi|$.
 \qed
\end{lemma}
\begin{proof}
Label $\mathbf{V}_\varphi = \{V_i\}_{1 \le i \le n}$ so that $V_1 \prec \dots \prec V_n$, and let $\delta = \bigwedge_{\alpha \in \Lintphi^-} \big([\alpha] \bigwedge_V \beta^{\alpha}_V\big)$.
Declare the answer to be ``no'' if there is any $\alpha$, $V$ such that $V \in \mathrm{dom}(\alpha)$\footnote{Here $\alpha$ represents an intervention considered as a partial function (see Def.~\ref{defn:scmintervention}).} and $\beta^{\alpha}_V \neq \alpha(V)$.
Construct a table whose rows are labeled by the elements of $\Lintphi^-$, and columns by the variables $V_1, \dots, V_n$.
Populate the cells by the following procedure: in row $\alpha \in \Lintphi^-$, column $V_i$, write $f_{V_i} \big(\beta^\alpha_{V_1}, \dots, \beta^\alpha_{V_{i-1}}\big) = \beta^\alpha_{V_i}$ if $V_i \notin \mathrm{dom}(\alpha)$ and otherwise leave the cell blank.
Two (non-blank) cells at the same column $V_i$ are said to be conflicting if
their functions map the same element of the domain to different elements of the range.
That is, two cells $f_{V_i}(b_1) = b_2$, $f_{V_i}(b'_1) = b'_2$ conflict if
$b_1 = b'_1$ and $b_2 \neq b'_2$.
Declare ``no'' if there is any column with two conflicting cells, and ``yes'' otherwise. The table contains $O\big(|\varphi|^2\big)$ cells.
Note that encoding the atom $V = v$ is $O\big(|\varphi|\big)$, and encoding $\cancel{\beta}_V$ takes $O\big(|\varphi|^2\big)$ bits. %as there are $O\big(|\varphi|^2\big)$ pairs $V = v$. ???
Thus the space required to store a cell is $O\big(|\varphi|^3\big)$, and checking for conflicts is in $\mathsf{P}$.

Soundness:
for each $V$, choose a fixed value $v^* \in \Val(V) \setminus \Val_\varphi(V)$ and for each $\alpha$ let $\gamma^\alpha_V = \beta^\alpha_V$ if $\beta^\alpha_V \in \Val_\varphi(V)$, and $\gamma^\alpha_V$ be $V = v^* \in \Lprop$ otherwise.
Replace every $\beta^\alpha_V$ appearing in the above table with $\gamma^\alpha_V$ and construct a deterministic model $\mathcal{F} = \{f_V\}_V$ such that $f_{V_i}$ satisfies all the constraints of column $V_i$; this is possible since no two cells in this column conflict.
We claim that $\mathcal{F} \models \delta$. %; by construction, $\mathcal{F} \models \bigwedge_{X \prec Y} Y \ninfluencesindirect X$.
Let $\alpha \in \Lintphi^-$; we show that $\mathcal{F} \models [\alpha] \bigwedge_V \beta^{\alpha}_V$.
The inductive hypothesis: $\mathcal{F} \models [\alpha] \gamma^{\alpha}_{V_1} \land \dots \land \gamma^{\alpha}_{V_i}$.
Base case: if $V_1 \in \mathrm{dom}(\alpha)$ then $\gamma^{\alpha}_{V_1}$ is $V_1 = \alpha(V_1)$ and if $V_1 \notin \mathrm{dom}(\alpha)$ then $\mathcal{F} \models \gamma^{\alpha}_{V_1}$ by construction of $\mathcal{F}$.
Inductive step: for each $1 \le i \le m$ we have that $\mathcal{F} \models [\alpha] \gamma^\alpha_{V_i}$, and we claim that $\mathcal{F} \models [\alpha] \gamma^\alpha_{V_{m+1}}$. This is trivial if $V_{m+1} \in \mathrm{dom}(\alpha)$ and if not, follows easily by construction of $\mathcal{F}$.

Completeness:
suppose $\delta \in \Delta_\prec$ and $\mathcal{F} = \{f_V\}_V \models \delta$.
Fix $\{v^*\}_V$ and define $\{\gamma^\alpha_V\}_{\alpha, V}$ as above and suppose without loss that for each $V \in \mathbf{V}_\varphi$, the image of $f_V$ is contained in $\Val_\varphi(V) \cup \{v^*\}$;\footnote{E.g., transform $\mathcal{F}$ to $\mathcal{F}'$ by remapping any outputs outside this set to $v^*$. This ensures that $\mathcal{F} \models [\alpha] \cancel{\beta}_V \leftrightarrow \mathcal{F}' \models [\alpha] (V = v^*)$ for any $V$, $\alpha$ so $\mathcal{F}' \models \delta$.} this entails that $\mathcal{F} \models [\alpha] \beta^\alpha_V \leftrightarrow [\alpha]\gamma^\alpha_V$.
Note that there cannot be any $\alpha$, $V$ such that $V \in \mathrm{dom}(\alpha)$ and $\beta^\alpha_V \neq \alpha(V)$ since this would mean $\mathcal{F} \not\models \delta$.
We claim that in each column $V_i$ of the table constructed from $\delta$, no two cells conflict.
Suppose toward contradiction the cells at rows $\alpha$, $\alpha'$ conflict.
% Let $j \ge 0$ be the largest index such that $\beta^{\alpha}_{V_k} = \beta^{\alpha'}_{V_k} \neq \cancel{\beta}_{V_k}$ for each $1 \le k \le j$.
% By relabeling we can assume that there is some $j \in \{1, \dots, i-1\}$ such that $\beta^\alpha_{V_k} = \beta^{\alpha'}_{V_k} = \cancel{\beta}_{V_k}$ iff $1 \le k \le j$.
Then we have that 
\begin{align*}
 \mathcal{F} \models [\gamma^{\alpha}_{V_1}, \dots, \gamma^{\alpha}_{V_{i-1}}] \gamma^{\alpha}_{V_i} \land [\gamma^{\alpha}_{V_1}, \dots, \gamma^{\alpha}_{V_{i-1}}] \gamma^{\alpha'}_{V_i}
\end{align*}
where $\gamma^{\alpha}_{V_i} \neq \gamma^{\alpha'}_{V_i}$, which is patently absurd.
\end{proof}
% \begin{corollary}
%  Given $\epsilon \in \Lfull$, deciding satisfiability is $\mathsf{NP}$-complete.
% \end{corollary}
% \begin{proof}
%  
% \end{proof}

\begin{theorem} 
  \label{thm:complexity}
  Each $\textsc{Prob-Causal-Sat}_i \in \mathsf{PSPACE}$.
\end{theorem}
\begin{proof}
Algorithm: for each order $\prec$ and subset $\Delta' \subset \Delta_\varphi$ of size $\left|\Delta'\right| \le |\varphi|$ such that $\Delta' \subset \Delta_\prec$,
form a formula $\varphi({\Delta'})$
in the existential theory of the reals ($\exists\mathbb{R}$), over variables $\{\mathbb{P}(\delta')\}_{\delta' \in \Delta'}$ as follows.
Conjoin the equations $\sum_{\delta'} \mathbb{P}(\delta') = 1$ and $\bigwedge_{\delta'}\mathbb{P}(\delta') \ge 0$ to the result of replacing any $\mathbb{P}(\epsilon)$ appearing in $\varphi$ with $\sum_{\delta' \models \epsilon} \mathbb{P}(\delta')$.
Then check satisfiability of $\varphi({\Delta'})$ via a $\mathsf{PSPACE}$ decision procedure for $\exists\mathbb{R}$ \cite{Canny}. Declare $\varphi$ sat. iff any $\varphi({\Delta'})$ is sat.
% Why does it work?
Completeness: if $\varphi$ is satisfiable, we have a model with a small $\Delta'$ by Lem. \ref{lem:smallmodel}, and the corresponding $\{\mathbb{P}(\delta')\}_{\delta'}$ witness $\varphi({\Delta'})$. Soundness: if $\varphi({\Delta'})$ is satisfiable in $\exists\mathbb{R}$ for some $\Delta'$, then apply Lem.~\ref{lem:metamodel:2} to get a model satisfying $\varphi$.
$\mathsf{PSPACE}$: note that storing an element of $\Delta_\varphi$ takes $O(|\varphi|^4)$ space and
checking that $\Delta' \subset \Delta_\prec$ is $\mathsf{P}$ by Lem.~\ref{lem:detlogic:p}.
% 
% Why is it $\mathsf{PSPACE}$?
% We proceed to bound the size of $\Delta_\varphi$.
% Let $n = \left|\varphi\right|$ and $m = |\mathbf{V}_\varphi| \le n$, label $\mathbf{V}_\varphi = \{V_1, \dots, V_m\}$, and let $x_i = |\Val_\varphi(V_i)|$ for each $1 \le i \le m$.
% Let $S = \sum_i x_i$ and $P = \prod_i x_i = \left|\Val_\varphi(\mathbf{V}_\varphi)\right|$.
% Since every element of $S$ corresponds to a value mentioned in $\varphi$, we have $S \in O(n)$. Thus there are $k, N$ such that $S \le k n$ for all $n \ge N$.
% By the arithmetic-geometric mean inequality, $S \ge m P^{1/m}$, so $k^m \left(\frac{n}{m}\right)^m \ge P$.
% Thus $P \in O(l^n)$ for some $l$.
% Now the number of formulas of the form $\bigwedge_{V} \beta^\alpha_V$ is $(x_1 + 1) \dots (x_m + 1) \in O(n l^n)$, and $\left|\Lintphi^-\right| \in O(n)$.
% Thus $|\Delta_\varphi| \in O\big(n^n l^{n^2}\big)$, and $O( n \lg n^n l^{n^2}) = O(n^2 \lg n + n^3) = O(n^3)$ bits are required to encode $n$ elements from $\Delta_\varphi$.
% How much space is needed to check that a $\delta \in \Delta_\varphi$ belongs to $\Delta_\prec$?
% The length of $\delta$ is $O\big(n \lg (x_1+1)\dots(x_m+1)\big) = O(n \lg P) = O(n^3)$ and there are $O(n^2)$ pairs in the conjunction $\bigwedge_{X \prec Y} Y \ninfluences X$.
% Satisfiability in $\Ll^+$ is $\mathsf{NP}$ and a fortiori $\mathsf{PSPACE}$ by Lem.~\ref{lem:detlogic:pspace}.
\end{proof}

\section{Probabilistic Programs}
Thm. \ref{thm:soundcomplete} establishes soundness and completeness for both the class of measurable SCMs, and also for the more restricted class of computable SCMs. The latter result is especially useful for establishing a link to an alternative perspective on causal modeling, emphasizing a procedural rather than declarative aspect \cite{AC17}. We take a probabilistic program to be any generative algorithm:
\begin{definition}[Probabilistic simulation model]
  A \emph{probabilistic simulation} is a probabilistic Turing machine with a read-only random bit tape, a work tape, and a write-only variable tape encoding endogenous variables $\mathbf{V}$.
\end{definition}
A probabilistic simulation outputs values for endogenous variables, eventually establishing a complete endogenous instantiation $\mathbf{v}$ on the variable tape. Thus, like a SCM, a probabilistic simulation model $\aT$ gives a probability distribution $P_{\aT}(\mathbf{V})$.
The following definition of intervention endows these models with a genuine causal interpretation:
\begin{definition} Given a computable intervention $i$ (as in Def. \ref{defn:scmintervention}) and a corresponding oracle for $i$, the simulation $i(\aT)$ emulates $\mathsf{T}$ but acts as if the square for any $X \in i$ is fixed to the value $i(X)$; it dovetails this emulation with a procedure that writes $i(X)$ to $X$ for all $X \in \text{dom}(i)$.
\end{definition}

In \cite{II2019} a subclass $\mathcal{T}^*$ of simulation programs is studied, namely those (1) that satisfy a strong ``functionality'' property for interventions, (2) that produce a solution under every intervention, and (3) for which the direct causal influence relation is well-founded (analogously to the requirement above for SCMs). We can then obtain:
\begin{theorem}
\label{thm:equivalence}
For every $\mathfrak{M} \in \mathcal{M}^*$ there is a $\mathsf{T} \in \mathcal{T}^*$ such that, for every computable intervention $i$, we have $P_{i(\mathfrak{M})}(\mathbf{V}) = P_{i(\aT)}(\mathbf{V})$, and vice versa.
\end{theorem} 
\begin{proof} The only difference between the present setting and the one in \cite{II2019} is the presence of an infinite sequence of random bits (recall Def. \ref{compmodel} for SCMs).
Thus, the same construction as in \cite[Thm. 1]{II2019} gives a model for which, when a random bit string is fixed, we have equivalence under any interventions of the endogenous variables.
\end{proof}
With the obvious interpretation of $\Ll_i$ formulas over simulation models in $\mathcal{T}^*$, we thus have:
\begin{corollary} Each $\AX_i$ is sound and complete for the validities of $\Ll_i$ with respect to  $\mathcal{T}^*$. \qed
\end{corollary}

\section{Conclusion}
We have introduced a series of increasingly expressive languages encoding levels of the ``ladder of causation'' \cite{Shpitser,Pearl2009}, interpreted over both standard structural causal models and probabilistic simulation programs. We moreover established some fundamental theoretical results about these languages, including finitary axiomatizations and $\mathsf{PSPACE}$ complexity upper bounds for satisfiability and validity. This marks the first systematic study of a probabilistic logic of causal counterfactuals.

Along the way we also noted how logical languages might help to illuminate aspects of probabilistic causal reasoning.
As a final illustration, let us return again to the $do$-calculus. We noted two seminal formula schemas from $\Ll_2$ that feature centrally in the $do$-calculus:
\begin{equation}\label{do1}\mathbb{P}\big([\mathbf{X} \wedge \mathbf{Z}]\mathbf{Y}|[\mathbf{X} \wedge \mathbf{Z}]\mathbf{W}\big)  \equiv   \mathbb{P}\big([\mathbf{X}]\mathbf{Y}|[\mathbf{X}](\mathbf{Z} \wedge \mathbf{W})\big)
\end{equation}
\begin{equation}\label{do2}\mathbb{P}\big([\mathbf{X} \wedge \mathbf{Z}]\mathbf{Y}|[\mathbf{X}\wedge \mathbf{Z}]\mathbf{W}\big)  \equiv  \mathbb{P}\big([\mathbf{X}]\mathbf{Y}|[\mathbf{X}]\mathbf{W}\big)
\end{equation}
Where $\Gamma$ is a set of $\Ll_3$ formulas, let $\Gamma_{do}$ be all instances of (\ref{do1}) and (\ref{do2}) that can be inferred from $\Gamma$ using the rules and axioms of $\AX_3$ (equiv. are entailed by $\Gamma$).
As usual, we  say $\Gamma \models \varphi$ to mean that $\Gamma$ semantically entails $\varphi$, while
$\Gamma \vdash \varphi$ means there is a proof in $\AX_3$ of $\varphi$ from assumptions in $\Gamma$. The completeness results of \cite{Huang,Shpitser}, together with Thm. \ref{thm:soundcomplete}, establish the following combined completeness result: \begin{corollary} \label{interestingcorollary} Let $\mathbf{p}$ be a term of $\Ll_1$ and $\epsilon \in \Ll_{2}^{\textnormal{base}}=\Ll_{\textnormal{cond}}$. Then $\Gamma \models \mathbb{P}(\epsilon) \equiv \mathbf{p}$ implies $\Gamma_{do} \vdash \mathbb{P}(\epsilon) \equiv \mathbf{p}$. \qed \end{corollary} In other words, to know whether a causal effect $\mathbb{P}(\epsilon)$ can be reduced to a pure probabilistic expression $\mathbf{p}$, given some set of assumptions $\Gamma$, it suffices to take as premises only instances of (\ref{do1}) and (\ref{do2}) (which by Prop. \ref{docalc} can all be inferred from specific graphical properties), and apply the calculus $\AX_3$ (or simply $\AX_2$). The derivation of (\ref{do}) from (\ref{second})--(\ref{fourth}) above is a concrete illustration of this Corollary.

More generally, we submit that the formalization of causal languages offered in this paper helps to clarify what exactly the levels of the causal hierarchy come to, and how we might gain a better understanding of how they relate to each other and to tasks that an intelligent agent might need to solve. For a start on such exploration, see \cite{BCII}.

\subsection{Future Work}
Our work opens up a number of possibilities for further investigation. We mention several here.

Although the complexity of decision problems for the languages considered here is relatively low compared to many expressive logical systems, it may be desirable to consider yet smaller fragments of probability logic. For example, the language of probability statements with linear inequalities remains in $\mathsf{NP}$ \cite{Fagin}. While this fragment is too impoverished to express general assertions about conditional probability, one could extend the language only minimally; cf. \cite{Ivanovska}.

It is also natural to consider more expressive languages. For instance, \cite{II2019} included a (deterministic) causal influence relation $\influences$ explicitly in the logical language, showing, e.g., that transitivity of this relation characterizes exactly the ``local'' SCMs \cite{Pearl2009}. Would this characterization extend to the probabilistic interpretation of causal influence? Indeed, from the perspective of causal learning and reasoning it would be natural to include explicit statements about the underlying graph in the logical syntax \cite{Geiger,Hyttinen1,Hyttinen,Trian}.

On the other hand, leaving graph properties merely implicit as we have done here raises numerous theoretical questions about graph definability, as briefly explored above. Analogous to the case of modal logic, we can ask for the class of graphical properties that can be defined by $\Ll_1$, $\Ll_2$, or $\Ll_3$. Considering different types of causal graphs may lead to variations of this question, e.g., with \emph{mixed ancestral graphs} \cite{Spirtes} which we know reveal a different version of Prop. \ref{docalc} \cite{Zhang08}.

Finally, while Props. \ref{l1l2} and \ref{l2l3} report known results establishing basic strictness of the hierarchy, it would of course be desirable to develop a much more comprehensive and systematic theory of expressiveness for the three languages, again akin to what we have for many other logical languages. What kinds of invariance properties do these languages imply? We leave these open questions for future work.

\section{Acknowledgments}
This material is based upon work supported by the National Science Foundation Graduate Research Fellowship Program under Grant No. DGE-1656518, and by the Center for the Study of Language and Information.

\bibliography{aaai20}

\newcommand{\SortNoop}[1]{}
\begin{thebibliography}{}

\bibitem[\protect\citeauthoryear{Ackerman, Freer, and Roy}{2019}]{Ackerman2019}
Ackerman, N.~L.; Freer, C.~E.; and Roy, D.~M.
\newblock 2019.
\newblock On the computability of conditional probability.
\newblock {\em Journal of the ACM} 66(3).

\bibitem[\protect\citeauthoryear{Avin, Shpitser, and Pearl}{2005}]{Avin}
Avin, C.; Shpitser, I.; and Pearl, J.
\newblock 2005.
\newblock Identifiability of path-specific effects.
\newblock In {\em Proceedings of the International Joint Conference on
  Artificial Intelligence (IJCAI)}.

\bibitem[\protect\citeauthoryear{Bareinboim \bgroup et al\mbox.\egroup
  }{2020}]{BCII}
Bareinboim, E.; Correa, J.~D.; Ibeling, D.; and Icard, T.
\newblock 2020.
\newblock {On Pearl's Hierarchy and the Foundations of Causal Inference}.
\newblock Technical Report R-60, Causal AI Lab, Columbia University.

\bibitem[\protect\citeauthoryear{{\SortNoop{Benthem}}van~Benthem}{2001}]{vanBenthem2001}
{\SortNoop{Benthem}}van~Benthem, J.
\newblock 2001.
\newblock Correspondence theory.
\newblock In Gabbay, D.~M., and Guenthner, F., eds., {\em Handbook of
  Philosophical Logic}, volume~3. Springer.

\bibitem[\protect\citeauthoryear{Bingham \bgroup et al\mbox.\egroup
  }{2019}]{Pyro}
Bingham, E.; Chen, J.~P.; Jankowiak, M.; Obermeter, F.; Pradhan, N.;
  Karaletsos, T.; Singh, R.; Szerlip, P.; Horsfall, P.; and Goodman, N.~D.
\newblock 2019.
\newblock Pyro: Deep universal probabilistic programming.
\newblock {\em Journal of Machine Learning Research} 28:1--6.

\bibitem[\protect\citeauthoryear{Canny}{1988}]{Canny}
Canny, J.
\newblock 1988.
\newblock Some algebraic and geometric computations in {PSPACE}.
\newblock In {\em Proceedings of the Twentieth Annual ACM Symposium on Theory
  of Computing}, STOC '88,  460--467.
\newblock New York, NY, USA: ACM.

\bibitem[\protect\citeauthoryear{Chv\'{a}tal}{1983}]{Chv1983}
Chv\'{a}tal, V.
\newblock 1983.
\newblock {\em Linear Programming}.
\newblock W.~H. Freeman and Co.

\bibitem[\protect\citeauthoryear{Fagin, Halpern, and Megiddo}{1990}]{Fagin}
Fagin, R.; Halpern, J.~Y.; and Megiddo, N.
\newblock 1990.
\newblock A logic for reasoning about probabilities.
\newblock {\em Information and Computation} 87:78--128.

\bibitem[\protect\citeauthoryear{Galles and Pearl}{1998}]{GallesPearl}
Galles, D., and Pearl, J.
\newblock 1998.
\newblock An axiomatic characterization of causal counterfactuals.
\newblock {\em Foundations of Science} 3(1):151--182.

\bibitem[\protect\citeauthoryear{Geiger and
  Meek}{1999}]{Geiger:1999:QES:2073796.2073822}
Geiger, D., and Meek, C.
\newblock 1999.
\newblock Quantifier elimination for statistical problems.
\newblock In {\em Proceedings of the Fifteenth Conference on Uncertainty in
  Artificial Intelligence}, UAI'99,  226--235.
\newblock San Francisco, CA, USA: Morgan Kaufmann Publishers Inc.

\bibitem[\protect\citeauthoryear{Geiger and Pearl}{1990}]{Geiger}
Geiger, D., and Pearl, J.
\newblock 1990.
\newblock On the logic of causal models.
\newblock {\em Machine Intelligence and Pattern Recognition} 9:3--14.

\bibitem[\protect\citeauthoryear{Halpern and Pearl}{2005}]{10.1093/bjps/axi147}
Halpern, J.~Y., and Pearl, J.
\newblock 2005.
\newblock {Causes and Explanations: A Structural-Model Approach. Part I:
  Causes}.
\newblock {\em The British Journal for the Philosophy of Science}
  56(4):843--887.

\bibitem[\protect\citeauthoryear{Halpern}{2000}]{Halpern2000}
Halpern, J.~Y.
\newblock 2000.
\newblock Axiomatizing causal reasoning.
\newblock {\em Journal of Artificial Intelligence Research} 12:317--337.

\bibitem[\protect\citeauthoryear{Huang and Valtorta}{2006}]{Huang}
Huang, Y., and Valtorta, M.
\newblock 2006.
\newblock Pearl's calculus of intervention is complete.
\newblock In {\em Proceedings of the 22nd Conference on Uncertainty in
  Artificial Intelligence (UAI)}.

\bibitem[\protect\citeauthoryear{Hyttinen, Eberhardt, and
  J\"{a}rvisalo}{2014}]{Hyttinen1}
Hyttinen, A.; Eberhardt, F.; and J\"{a}rvisalo, M.
\newblock 2014.
\newblock Constraint-based causal discovery: Conflict resolution with answer
  set programming.
\newblock In {\em Proceedings of the Thirtieth Conference on Uncertainty in
  Artificial Intelligence (UAI)}.

\bibitem[\protect\citeauthoryear{Hyttinen, Eberhardt, and
  J\"{a}rvisalo}{2015}]{Hyttinen}
Hyttinen, A.; Eberhardt, F.; and J\"{a}rvisalo, M.
\newblock 2015.
\newblock Do-calculus when the true graph is unknown.
\newblock In {\em Proceedings of the Thirty-First Conference on Uncertainty in
  Artificial Intelligence (UAI)}.

\bibitem[\protect\citeauthoryear{Ibeling and Icard}{2018}]{II2018}
Ibeling, D., and Icard, T.
\newblock 2018.
\newblock On the conditional logic of simulation models.
\newblock In {\em Proceedings of the 27th International Joint Conference on
  Artificial Intelligence (IJCAI 2018)}.

\bibitem[\protect\citeauthoryear{Ibeling and Icard}{2019}]{II2019}
Ibeling, D., and Icard, T.
\newblock 2019.
\newblock On open-universe causal reasoning.
\newblock In {\em Proceedings of the Conference on Uncertainty in Artificial
  Intelligence (UAI)}.

\bibitem[\protect\citeauthoryear{Ibeling}{2018}]{Ibeling18}
Ibeling, D.
\newblock 2018.
\newblock Causal modeling with probabilistic simulation models.
\newblock In {\em Proceedings of the 5th International Workshop on
  Probabilistic Logic Programming (PLP)},  36--48.

\bibitem[\protect\citeauthoryear{Icard}{2017}]{AC17}
Icard, T.~F.
\newblock 2017.
\newblock From programs to causal models.
\newblock In Cremers, A.; van Gessel, T.; and Roelofsen, F., eds., {\em
  Proceedings of the 21st Amsterdam Colloquium},  35--44.

\bibitem[\protect\citeauthoryear{Ivanovska and Giese}{2010}]{Ivanovska}
Ivanovska, M., and Giese, M.
\newblock 2010.
\newblock Probabilistic logics with conditional independence formulae.
\newblock In {\em 19th European Conference on Artificial Intelligence (ECAI)}.

\bibitem[\protect\citeauthoryear{Janzing and Sch\"{o}lkopf}{2010}]{IM}
Janzing, D., and Sch\"{o}lkopf, B.
\newblock 2010.
\newblock Causal inference using the algorithmic {M}arkov condition.
\newblock {\em IEEE Transactions on Information Theory} 56(10):5168--5194.

\bibitem[\protect\citeauthoryear{Lake \bgroup et al\mbox.\egroup
  }{2017}]{Lake2017}
Lake, B.~M.; Ullman, T.~D.; Tenenbaum, J.~B.; and Gershman, S.~J.
\newblock 2017.
\newblock Building machines that learn and think like people.
\newblock {\em Behavioral and Brain Sciences} 40.

\bibitem[\protect\citeauthoryear{Pearl and Bareinboim}{2012}]{Bareinboim}
Pearl, J., and Bareinboim, E.
\newblock 2012.
\newblock External validity: From do-calculus to transportability across
  populations.
\newblock {\em Statistical Science} 29(4):579--595.

\bibitem[\protect\citeauthoryear{Pearl}{1995}]{Pearl1995}
Pearl, J.
\newblock 1995.
\newblock Causal diagrams for empirical research.
\newblock {\em Biometrika} 82(4):669--710.

\bibitem[\protect\citeauthoryear{Pearl}{2009}]{Pearl2009}
Pearl, J.
\newblock 2009.
\newblock {\em Causality}.
\newblock CUP.

\bibitem[\protect\citeauthoryear{Perovi\'{c} \bgroup et al\mbox.\egroup
  }{2008}]{Perovic}
Perovi\'{c}, A.; Ognjanovi\'{c}, Z.; Ra\v{s}kovi\'{c}, M.; and Markovi\'{c}, Z.
\newblock 2008.
\newblock A probabilistic logic with polynomial weight formulas.
\newblock In {\em International Symposium on Foundations of Information and
  Knowledge Systems},  239--252.

\bibitem[\protect\citeauthoryear{Shpitser and Pearl}{2008}]{Shpitser}
Shpitser, I., and Pearl, J.
\newblock 2008.
\newblock Complete identification methods for the causal hierarchy.
\newblock {\em Journal of Machine Learning Research} 9:1941--1979.

\bibitem[\protect\citeauthoryear{Spirtes, Glymour, and
  Scheines}{2000}]{Spirtes}
Spirtes, P.; Glymour, C.; and Scheines, R.
\newblock 2000.
\newblock {\em Causation, Prediction, and Search}.
\newblock MIT Press.

\bibitem[\protect\citeauthoryear{Stengle}{1974}]{Stengle1974}
Stengle, G.
\newblock 1974.
\newblock A {N}ullstellensatz and a {P}ositivstellensatz in semialgebraic
  geometry.
\newblock {\em Mathematische Annalen} 207(2):87--97.

\bibitem[\protect\citeauthoryear{Tarski}{1949}]{Tarski1949-TARADM-3}
Tarski, A.
\newblock 1949.
\newblock A decision method for elementary algebra and geometry.
\newblock {\em Journal of Symbolic Logic} 14(3):188--188.

\bibitem[\protect\citeauthoryear{Tavares \bgroup et al\mbox.\egroup
  }{2019}]{Tavares}
Tavares, Z.; Koppel, J.; Zhang, X.; and Solar-Lezama, A.
\newblock 2019.
\newblock A language for counterfactual generative models.
\newblock www.zenna.org/publications/causal.pdf.

\bibitem[\protect\citeauthoryear{Tian, Kang, and Pearl}{2006}]{TianKP06}
Tian, J.; Kang, C.; and Pearl, J.
\newblock 2006.
\newblock A characterization of interventional distributions in
  semi-{M}arkovian causal models.
\newblock In {\em Proceedings of the AAAI Conference on Artificial Intelligence
  (AAAI)}.

\bibitem[\protect\citeauthoryear{Triantafillou and Tsamardinos}{2015}]{Trian}
Triantafillou, S., and Tsamardinos, I.
\newblock 2015.
\newblock Constraint-based causal discovery from multiple interventions over
  overlapping variable sets.
\newblock {\em Journal of Machine Learning Research} 16:2147--2205.

\bibitem[\protect\citeauthoryear{Weihrauch}{2000}]{Weihrauch}
Weihrauch, K.
\newblock 2000.
\newblock {\em Computable Analysis}.
\newblock Springer Verlag.

\bibitem[\protect\citeauthoryear{Zhang}{2008}]{Zhang08}
Zhang, J.
\newblock 2008.
\newblock Causal reasoning with ancestral graphs.
\newblock {\em Journal of Machine Learning Research} 9:1437--1474.

\bibitem[\protect\citeauthoryear{Zhang}{2013}]{Zhang}
Zhang, J.
\newblock 2013.
\newblock A {L}ewisian logic of causal counterfactuals.
\newblock {\em Minds and Machines} 23:77--93.

\end{thebibliography}
\bibliographystyle{aaai}

\end{document}